\documentclass[a4paper]{article}

\usepackage{graphicx}

 \usepackage{epsfig}
\usepackage{amssymb}
\usepackage{graphicx}
\usepackage{multirow}
\usepackage{subcaption}
\usepackage{amsmath}
\usepackage{xspace}
\usepackage{tabularx}
\usepackage[sort]{cite}
\usepackage{fullpage}
\usepackage{enumitem}
\usepackage{algorithm}
\usepackage{amsthm}
\usepackage{hyperref}
\usepackage{framed}
\usepackage{xcolor}
\usepackage[noend]{algpseudocode}

\algblockdefx{AtState}{END}[1]{\underline{In state {\sf #1}:}}{}
\algblockdefx{AtStateBreakLine}{END}[2]{ \underline{In state {\sf #1}} \newline  \underline{{\sf #2}:} }{}
\algcblockdefx{AtState}{AtState}{END}[1]{\underline{In state {\sf #1}:}}{\vspace{-1.5em}}{}
\algcblockdefx{AtStateBreakLine}{AtStateBreakLine}{END}[2]{ \underline{In state {\sf #1}} \newline  \underline{{\sf #2}:} }{}

\newcommand{\G}{\mathcal{G}}
\newcommand{\TVG}{\ensuremath{\G=(V,E, \mathbb{T},\rho)}}

\newcommand{\Explore}{\textsc{Explore}\xspace}

\newcommand{\CExplore}{\textsc{CautiousExplore}\xspace}

\newcommand{\sBounce}{\textsf{Bounce}\xspace}

\newcommand{\sReturn}{\textsf{Return}\xspace}

\newcommand{\pMeeting}{\textsl{meeting}}

\newcommand{\pSees}{\textsl{sees}}
\newcommand{\dLeft}{\textit{left}\xspace}
\newcommand{\dRight}{\textit{right}\xspace}
\newcommand{\dNil}{\textit{nil}\xspace}
\newcommand{\vTtime}{\ensuremath{Ttime}\xspace}
\newcommand{\vTsteps}{\ensuremath{Tnodes}\xspace}

\newcommand{\vEtime}{\ensuremath{Etime}\xspace}
\newcommand{\vEsteps}{\ensuremath{Enodes}\xspace}

\newcommand{\vMtime}{\ensuremath{EMtime}\xspace}

\newcommand{\BH}{{\sc Bh}\xspace}
\newcommand{\BHS}{{\sc Bhs}\xspace}

\newtheorem{observation}{Observation}

\newcommand{\ic}{\ic}

\newcommand{\tc}{\tc}
\newcommand{\BHP}{{\sc BHS-Problem}\xspace}

\newcommand{\R}{{\sc Retroguard}\xspace}

\newcommand{\A}{{\sc Avanguard}\xspace}
\newcommand{\Leader}{{\sc Leader}\xspace}
\newcommand{\RT}{{\sf CautiousPendulum}\xspace}
\newcommand{\DP}{{\sf CautiousDoubleOscillation}\xspace}
\newcommand{\GL}{{\sf Gather\&Locate}\xspace}

\newtheorem{theorem}{Theorem}

\newtheorem{lemma}[theorem]{Lemma}
\newtheorem{definition}{Definition}

\title{Tight Bounds for Black Hole Search in Dynamic Rings} 

\author{Giuseppe Antonio Di Luna$^{\dag}$, Paola Flocchini$^{*}$, Giuseppe Prencipe$^{\ddag}$, Nicola Santoro$^{\&}$ \\
\small
$\dag$:Sapienza University, $*$: University of Ottawa, $\ddag$: University of Pisa, $\&$: University of Carleton. 
} 
\date{}

\begin{document}

\maketitle
\begin{abstract}
In this paper, we start the investigation of distributed computing by mobile agents in dangerous dynamic networks.
The danger is posed by the presence in the network of a {\em black hole} (\BH), a harmful site that destroys all incoming agents without leaving any trace. The problem of determining the location of the black hole in a network,
known as {\em black hole search} (\BHS), has been extensively studied in the literature,
 but always and only assuming
that the network is static. 
At the same time, the existing results on mobile agents computing in dynamic networks
never consider the presence of harmful sites. 
In this paper we start filling this research gap by studying black hole search in temporal rings, specifically
focusing on {\em 1-interval connectivity} adversarial dynamics.

Clearly the task is dangerous for the agents, as any agent entering the \BH\ will be destroyed; the problem
is solved if within finite time at least one agent survives and knows the location of \BH.
The main complexity parameter of \BHS\ is the number of agents (called {\em size}) needed to solve
the problem; another important parameter is the number of moves (called {\em cost}) performed by the agents; in synchronous
systems, such as temporal rings, an additional complexity measure is the amount of {\em time}
until termination occurs.

Feasibility and complexity depend on many parameters; in particular: whether the agents 
start from the same safe node or from possibly distinct safe locations,
the size $n$ of the ring, whether or not $n$ is known, and the type of inter-agent communication
(whiteboards, tokens, face-to-face, visual).
In this paper, we provide a {\em complete} feasibility characterization
for all instances of those parameters; all our algorithms are size optimal.
Furthermore, we establish lower bounds on the cost (i.e., the number of moves) and time of  
size-optimal solutions for all instances of those parameters and show that our algorithms achieve those bound.\\

{\bf keywords}: Mobile agents, black hole search, dynamic ring. 
\end{abstract}

\bigskip


\section{Introduction}

\subsection{Background}

%
%
%
When computing in networked environments, 
{\em mobile agents} are used both as a theoretical computational
paradigm and as a system-supported programming platform.
These distributed mobile computing environments
are subject to
several security threats, including those posed by 
 {\em host attacks}; that is, the presence in a site of processes 
 that harm incoming agents.
A particularly dangerous host, called  {\em black
hole} (\BH), is a network site hosting a stationary
process that
disposes of visiting agents upon their arrival, leaving  no
observable trace of such a destruction.  
Notice that the existence of a 
 black hole  is not uncommon in networked systems supporting code mobility;  
 for example, both  the presence  of a  virus that trashes any incoming message (e.g., by classifying it as spam) 
and the undetectable crash failure of a site  render that site a black hole. 
 Clearly, in presence of such a harmful host, the first step must be to determine  its location.
Black Hole Search (\BHS) is the distributed problem of determining the location of the black hole by a team of system agents. The problem, also called dangerous graph exploration, is solved if, within finite
time,  at least one agent survives   knowing the location of the black hole  \cite{DobrevFPS07}.

The task to identify the \BH\ is clearly
dangerous for the searching agents and might be impossible to perform.
The research concern has been to determine under what conditions a team
of mobile agents can successfully accomplish this task.
 \BHS has been extensively investigated  in a variety of settings,  depending on the types of communication mechanisms employed by the agents, their level of synchronicity, the topology of the network, etc., and under a variety of assumptions on the agents' knowledge and capabilities
  (e.g., \cite{ChDLM11b,ChDLM13,CzKMP06,Dobrev2002,DobrevFPS07,FIS12,KMRS07,KMRS08}; for a recent survey, see \cite{MaS19}).
All these investigations share a common trait: the dangerous network on which the agents operate is {\em static}  
 i.e., the link structure does not change in time.

Recently, research within  distributed computing  has  started to focus on 
mobile-agent computing in  {\em time-varying  graphs} (a.k.a., highly dynamic graphs),  
 i.e., graphs where the topological changes 
  are not limited to  sporadic and disruptive events (such as process failures, links congestion, etc), 
  but are rather inherent in the nature of the network (e.g., see {\cite{CaFQS12}).
 A large body of literature exists on time-varying graphs as they model a wide range of 
 modern networked systems whose dynamic nature  is the natural product of  innovations 
 in communication technology (e.g., wireless networks), in software layer (e.g., a controller in a software defined network), and in society (e.g., the pervasive nature of smart mobile devices). 
The vast majority of the research on dynamic networks 
has considered  time-varying graphs with {\em discrete} temporal dynamics; that is,
the network is seen as an infinite sequence of static graphs with the same vertex set, 
and it is usually called {\em evolving}  or {\em temporal} graph.
The study of mobile agents in temporal graphs 
includes both centralized and distributed investigations
(e.g., \cite{IlcinkasW18,ErHK15,AaKM14,IlKW14,Michail2014}). 
Notice that, in such temporal graphs, the distributed computation is by definition {\em synchronous};
extensive investigations have been carried out under specific  assumptions on the discrete temporal dynamics,
including
the minimal assumption of {\em temporal connectivity} (e.g., \cite{BournatDD16,CaFMS14,GoFMS19}), the popular assumption 
of {\em 1-interval connectivity} (and its generalization of {\em T-interval connectivity}) (e.g., \cite{AbsM14,Agarwalla2018,briefdiluna,DiLFPPSV20,HaeK12,KuhLyO10,KuhLoO11,OdW05}), 
and  {\em periodicity} (e.g., \cite{AaKM14,CaFMS14,ErS20,FlMS13,JatYG14}).
 While the study of mobile agents 
on static networks is really mature, and generated a copious literature (e.g., see \cite{FlPS19}  and chapters therein),
the research on mobile agents in temporal graphs 
is still in its infancy, especially from a purely distributed perspective.
Its  focus has been on   classical
problems, such as graph exploration \cite{GoFMS19,BournatDP17,DiLDFS20}, gathering \cite{DiLFPPSV20,BournatDP18}, scattering \cite{Agarwalla2018}, and  {grouping} \cite{DaDPP19} 
under a variety of settings,  depending on the types of communication mechanisms employed by the agents, 
 the topology of the network, etc., and under a variety of assumptions on the agents' knowledge and 
 capabilities   (for a recent detailed survey  see \cite{DiL19}). 
In spite of the different settings, these investigations  share a common trait: they assume that the 
 dynamic network on which the agents operate is {\em safe}; 
 i.e., there is no  \BH.

Summarizing, practically nothing is known on distributed computing by mobile agents in 
{\em dangerous dynamic networks}.  In this paper, we start filling this research gap.

\subsection{Problem}

 In this paper, we  study \BHS\  in  a  {\em temporal  ring}
  under 
 the {\em 1-interval connectivity} adversarial dynamics.
 In other words, the network is a  synchronous ring where one of the nodes is 
 a \BH\ and, at any time unit,  one edge 
 (chosen by an adversary) is possibly missing.
 The problem to be solved  is to identify the location  of the \BH.
 The problem is solved by
a  team of mobile agents, executing the same protocol and initially deployed at safe node(s) in the network,
 if within finite time at least one agent survives and unambiguosly knows
 the location  of the \BH. 
 
 The main research questions is to determine the minimum number of
 agents needed to solve \BHS; this parameter is
  called {\em team size} or simply {\em size}.
 Another important complexity measures is the number of moves, 
 called {\em cost},  performed by the agents; 
 in synchronous
systems, such as temporal rings, an additional complexity measure is the amount of {\em time}
until termination occurs.

Feasibility and complexity depend on many parameters; in primis,  whether the agents 
start from the same safe node ({\em colocated}) or from possibly distinct safe locations ({\em scattered}),
on the size $n$ of the ring, on whether or not $n$ is known, and whether the agents have distinct ids or are
anonymous. 
A factor that is particularly important is the mechanism provided to the
agents to communicate and interact. 
In the  literature on distributed computing by mobile agents
different  models of interaction and communication  with different
capabilities  have been considered. Listed in increasing computational power, 
these models are: {\em Whiteboard}, whereas each node provides all
visiting nodes with a shared memory, called whiteboard, that can be
accessed, in fair mutual exclusion, to exchange information;
{\em Pebble} (or {\em Token}), whereas each agent has available
 a pebble that can be carried and,  when at a node,
 can be placed there or taken from there,
 the last two operations performed in fair mutual exclusion;
 {\em FaceToFace} (F2F),  where the agents can exchange information 
 only when they are in the same node at the same time;
 and  {\em Vision}, where an agent can only sense the  other agents 
 in the same node at the same time but cannot explicitly communicate
 with them.
 These models can be conveniently grouped into two classes:
   {\em endogenous} (or internal), where the agents rely only on their own internal capabilities to communicate and can do so only when   present on the same node (F2F and   Vision models);  
and   {\em exogenous} (or external), where the  agents make use of external tools (pebbles
and whiteboards) that allow them to leave traces or messages in the nodes of the network. 
Not surprisingly, 
 neither the solutions  for exploration of safe 
synchronous rings nor  the ones for \BHS\ devised for static synchronous rings can be applied
for the  exploration of a dangerous temporal ring.
Furthermore,  the existing trivial lower-bound for  \BHS\  in static synchronous rings,
that more than one agent is needed, does not provide any insight into the
complexity of the problem under investigation nor of the 
computational impact of the different parameters.

In other words, prior to this work, the feasibility and complexity of exploring
a dangerous 1-interval connected ring is an unexplored problem.

\subsection{Contributions}

In this paper, we investigate the black hole search problem  in an oriented 1-interval connected rings 
of size $n$ by a team of $k$ agents. 
We consider  both
  colocated and scattered agents, the communication mechanism that they employ,
   whether or not $n$ is known, and whether or not the agents  are
anonymous.  For each possible setting,
we   provide  a complete feasibility characterization.
Furthermore, whenever the problem is solvable,
we  establish tight bounds on  cost and time of 
a size-optimal solution.

We start   by showing  that   knowledge of $n$ is necessary for teams of any size $k$ and irrespectively of the other parameters; and, that $k < 3$ agents cannot solve the problem even in the strongest possible model (colocated agents,  and nodes   equipped with whiteboards).

 When the agents are colocated (i.e., start from the same   node), 
we   show that any optimal-size algorithm that uses endogenous communication
 requires $\Omega(n^2)$ cost and time  even in the strongest of the endogenous mechanisms, FaceToFace; 
 and we constructively prove that this bound is tight by designing a solution for the weakest of the endogenous models, Vision,  that has a move and time complexity of $\Theta(n^2)$.  With the more powerful exogenous mechanisms, we show a tight  bound of  ${\Theta}(n^{1.5})$: the lower bound holding for the strongest  Whiteboard model, the matching upper bound for the weakest one, Pebble.

When the agents are {\em scattered}   (i.e., start from distinct locations), we first observe that no optimal size algorithm can locate the black hole by using endogenous communication. 
 We then show that the scattering of agents impacts the cost complexity of optimal-size algorithms: any solution in this model has to pay a cost of $\Omega(n^2)$ rounds and moves. Also in this case the bound is tight; in fact,
we present a $\Theta(n^2)$ matching optimal-size solution algorithm for the weakest 
of the exogenous models (Pebble). 
A summary of the results is shown in Table \ref{table:res}.
\color{black}
\begin{table}
\begin{center}

\begin{tabular}{ |c|c|c|c|c| }

 \hline
 	 & \multicolumn{2}{c}{{\em Exogenous}} &  \multicolumn{2}{c|}{{\em Endogenous} } \\ \hline
	 	 &   Anonymous & IDs & Anonymous & IDs\\
 \hline  
 	{\em Colocated} & \multicolumn{2}{c|}{${\Theta}(n^{1.5})$ } &     impossible &  ${\Theta}(n^2)$\\ \hline
         {\em Scattered} & \multicolumn{2}{c|}{ ${\Theta}(n^2)$ } &  \multicolumn{2}{c|}{impossible}  \\ \hline

 \hline
\end{tabular}
\end{center}

\caption{Map of the results. \label{table:res}}
\end{table}

\subsection{Related Work}
 The existing  literature related to our research can be divided between 
that  considering  \BHS in static networks and   that  investigating  distributed 
computing by mobile agents in safe dynamic graphs. 

\smallskip
\noindent {\bf Agents in dangerous static  graphs.} 
The black hole search problem has been introduced by Dobrev {\em et al.} in their seminal paper \cite{Dobrev2002}. A panoply of papers followed \cite{Flocchini2009,balamohan2014exploring,d2013exploring,markou2012identifying,shi2018Token} solving the problem in different classes of graphs (trees \cite{CzKMP07}, rings and tori \cite{dobrev2007locating,opodis12,ChDLM11b}, and  in  graphs of arbitrary and possibly unknow topology \cite{Dobrev2002,dobrev2013exploring,CzKMP06}),  under several assumptions (see the recent survey \cite{MaS19}).   

The most relevant papers for our work are the ones investigating  the \BHS\ in static ring networks. 
In the {\em asynchronous} setting,
 optimal size and cost bounds have been established, 
   solving the problem with two  colocated agents and   $\Theta(n \log n)$   moves,
    in the whiteboard model  \cite{DobrevFPS07}, 
and subsequently in the   pebble model \cite{FIS12}. 
In \cite{ChDLM13}, it has been shown that two scattered agents with pebbles are sufficient to find a black hole on   
oriented rings with ${\cal O}(n \log n)$ moves \cite{4228188} 
and with    ${\cal O}(n^2)$   in unoriented rings in  \cite{dobrev2007locating}. 
In the {\em synchronous} setting, on the other hand, it is well known that two colocated non-anonymous agents with FaceToFace communication can solve the problem in arbitrary known graphs (and therefore on the ring) \cite{CzKMP06}. Finally, \BHS\ by scattered agents with constant memory
has been studied in \cite{ChDLM13}: in unoriented rings $3$ agents are necessary and sufficient  when equipped with movable tokens, while more agents are needed when tokens are not movable. 
All papers on black hole search  assume a static topology;
the only exception is the study of {\em carrier graphs}, a particular class of periodic graphs 
defined by circular intersecting routes  of public carriers, where the stops are the nodes of the graph and  the agents can board and disembark from a carrier at any stop \cite{FlKMS12}.

\smallskip
\noindent {\bf Agents in safe dynamic graphs.} 
The study of mobile agents in temporal graphs is rather recent and
includes both centralized and distributed investigations
(e.g., \cite{IlcinkasW18,ErHK15,AaKM14,IlKW14,Michail2014}). 
From a distributed perspective, the research has  so-far
considered three types of temporal dynamics:
\noindent{\em periodic} \cite{FlMS13},
{\em temporal connectivity} \cite{BournatDD16,BournatDP17,GoFMS19}, and 
{\em 1-interval connectivity}   \cite{Agarwalla2018,DaDPP19,DiLDFS20,DiLFPPSV20,GoFMS19};
 for a recent detailed survey  see \cite{DiL19}. 
Several of these works considered ring networks.

Specifically,  in {\em 1-interval connected  rings}: 
the  {\em gathering} problem has been investigated in \cite{DiLFPPSV20};
The {\em exploration} problem  by a set of anonymous agents has been studied in   \cite{DiLDFS20}  under several assumptions (handedness agreement, synchrony vs semi-synchrony, knowledge of $n$ vs landmark, etc...); a recent preprint \cite{m2020live} has closed some questions left-open by \cite{DiLDFS20} regarding the terminating exploration by a team of $3$ agents.
Always in 1-interval connected ring, recent papers investigated the problems of {\em grouping} \cite{DaDPP19} 
and {\em scattering} \cite{Agarwalla2018}.  Exploration of a {\em temporally connected ring} was examined in 
 \cite{BournatDD16,BournatDP17}. 

\section{Model and Preliminaries}

\subsection{The Model and the Problem}

The system   is modeled  as a  time-varying graph  \TVG, where  $V$ is a set of nodes, $E$ is a set of edges, $ \mathbb{T}$  is the temporal domain, and 
  $\rho: E \times \mathbb{T} \rightarrow \{0,1\}$, called {\em presence function}, indicates whether a given edge    is available at a given time \cite{CaFQS12}.

The graph $G=(V,E)$ is called {\em underlying} graph (or {\em footprint})  of $\G$.
  In this paper we consider {\em discrete} time; that is,  $ \mathbb{T} = \mathbb{Z}^+$.
Since time is  discrete,  the dynamics of the system  can be viewed  as
 a sequence of static graphs: ${\G} = G_0,G_1, \ldots,G_r, \ldots$, where $G_r =(V_r,E_r)$ is the   graph of the edges present at round $r$ (also called {\em  snapshot} at time $r$).  
 The time-varying graph  in this case is called  {\em temporal graph}  (or {\em evolving graph}). 
We use the term {\em adversary} to refer to the scheduler that decides the sequence of dynamic graphs.

A  temporal graphs  where connectivity is  guaranteed at every round  is called   {\em 1-interval connected};
that is, 
a temporal graph ${\cal G}$ is  1-interval connected (or always connected) if      $\forall G_i \in \mathcal{\G}$, $G_i$ is connected.  
 In this paper we focus on {\em dynamic rings}, defined as 1-interval connected temporal graphs whose footprint is a ring.
Let   ${\cal R} = (v_0, v_1, \ldots v_{n-1})$  be a  dynamic oriented ring, i.e., where each node $v_i$ has two ports, consistently labelled left and right connecting it to $v_{i-1}$ and $v_{i+1}$, respectively (all operations on the 
indices are modulo $n$).
A set    $A = \{ a_0, a_1, \dots, a_{k-1} \}$    of   mobile    agents    operate in  ${\cal R}$.  If the agents are initially   in the same node (called {\em home-base}),  we say that they are {\em colocated}; if they start from distinct arbitrary locations, we say that they are   {\em scattered}. 
When the agents are identical (i.e., do not have  distinct identifiers), we say that they are  {\em anonymous}. In case agents are not anonymous we assume that their identities are visibile (e.g., if several agents are on the same node they can see who is who).
 The agents  can move from node to neighbouring node, have bounded storage (${\cal O}(\log n)$ bits of internal memory suffice for our algorithms), have computing capabilities  and obey the same set of rules (i.e., execute the same algorithm).  
The agents operate in synchronous rounds, and they are all activated in  each round. Upon activation, an agent on node $v$ at round $r$
   takes  a local snapshot of $v$ that  contains  the set $E_r(v)$ of edges incident on $v$ at this round,  and  the set of agents present in $v$. The agent also interacts    with the other agents either explicitly or implicitly  (the method of interaction depends 
   on the communication mechanism employed and will be discussed later).
 On the basis of   the snapshot, the local interaction,  and the content of its local memory,
  an agent  then decides what action to take.  The action consists of a {\em communication step} (defined below) and a {\em move step}. In the move step the agent may decide to stay still or to move on an edge    $e=(v,v') \in E_r(v)$. In the latter case,  the agent will reach $v'$ in round $r +1$.  
 
The interaction among the agents is regulated by different {\em  communication mechanisms} depending on the model. 
 We consider two classes of communication mechanisms ({\em endogenous} and {\em exogenous}) which give rise to   four models.

\noindent {\bf Endogenous Mechanisms}
 rely only on the robots' capabilities without requiring  any external object. Among those we distinguish:
 
\noindent -  {\em Vision:}  the agents have no  explicit means of communication; they can only see each other when they reside on the same node. 

\noindent  - {\em FaceToFace (F2F): }  the agents can explicitly communicate among themselves  only when they reside on the same node.  
 \smallbreak
 
\noindent {\bf Exogenous Mechanisms}
do   require  external objects for the robots to exchange information. Among those we distinguish:

\noindent - {\em Pebble:} each agent is endowed with a single pebble that can be placed on or taken from a node.  On each node,  the concurrent actions of placing or taking  pebbles are done in  fair mutual exclusion.

\noindent  -{\em Whiteboard:}   each node contains a local shared memory, called
whiteboard, of size $O(\log n)$ where agents can write on  and read from.  Access to the whiteboard is done in adversarial but fair mutual exclusion.

Notice that  the mutual exclusion nature of 
the Pebble and Whiteboard models 
 allows  anonymous colocated agents to break the symmetry and assume different Ids.     


The temporal graph $\G$ contains a {\em black hole} (\BH), a node 
that destroys any incoming agent without leaving any detectable trace of that
destruction. The goal of a {\em black hole search} algorithm ${\cal A}$ is
to identify the location of the black hole,
that is:

%
%

\begin{definition}(\BHS)
Given a dynamic ring ${\cal R}$, and an algorithm ${\cal A}$ for a set of agents we say that ${\cal A}$ solves the \BHS if   at least one agent survives and terminates. Each agent that terminates hast to know the footprint  of ${\cal R}$  with the indication of  the location of the backhole. 
\end{definition}

The main measure of complexity is the number of agents, called {\em size}, used by the protocol.
The other important cost measures are the total number of moves performed by the agents, which we shall call {\em cost}, and {\em time} it takes to complete the task. 
 
In Figure \ref{fig:p0} are shown (a) four rounds of an execution in a dangerous dynamic ring,
 and (b) the space diagram representation that we will use in this paper.  
 The agent is represented as the black quadrilateral and it is moving clockwise; the  \BH\ is the black node. At round $r=2$ and $r=3$ the agent is blocked by  the missing edge. In the diagram,   the movement of the agent is represented as a solid line.

 \begin{figure}[tbh]
\center
  \begin{subfigure}[b]{0.48\textwidth}
  \center
    \includegraphics[width=\textwidth]{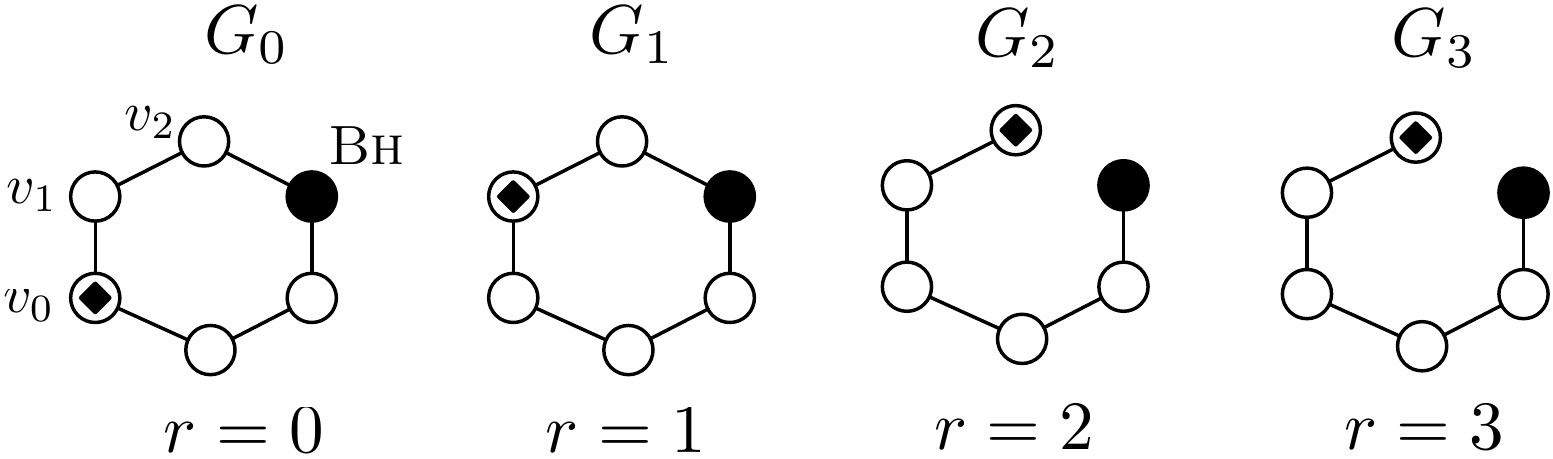}
    \label{graph1}
  \end{subfigure}
  \,\,
  \begin{subfigure}[b]{0.48\textwidth}
  \center
    \includegraphics[width=\textwidth]{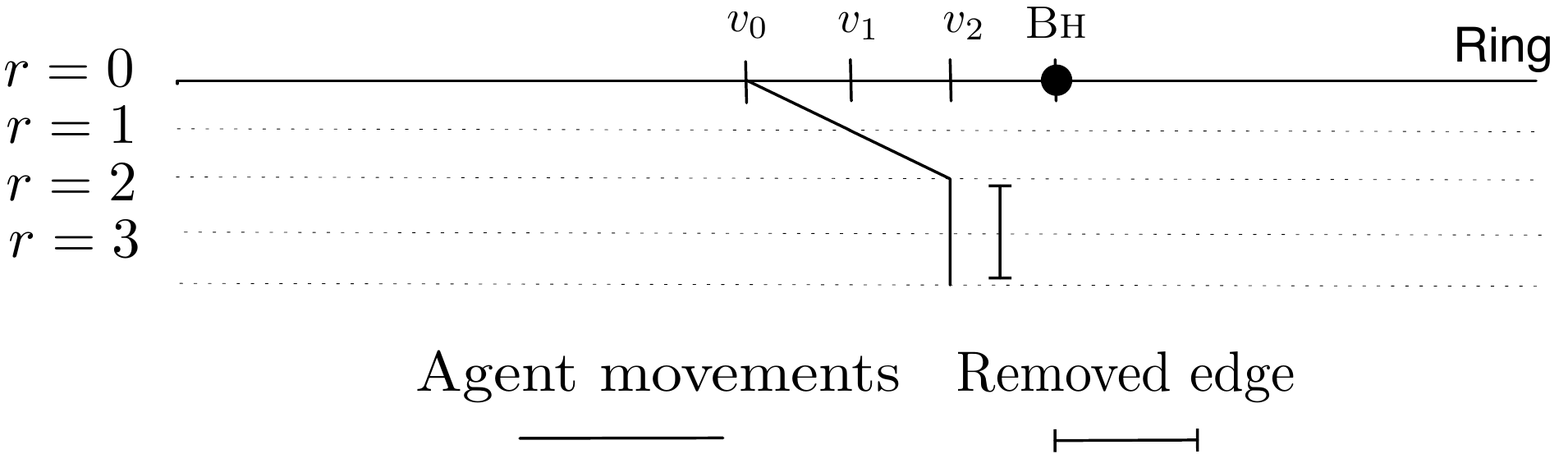}
    \label{graph2}
  \end{subfigure}
  \caption{(a) Execution in a dangerous dynamic ring, and 
  (b) its space diagram representation.}\label{fig:p0}
\end{figure}
\vspace{-0.5cm}

\section{Impossibilities and Basic Limitations}

In this section we show some general impossibility results that hold in the whiteboard model and hence under all communication mechanisms considered in this paper. More precisely, we establish that  three agents are necessary to locate the \BH and that, irrespectively of the number of agents available, the size of the ring must be known.

\begin{lemma}\label{lm:technical}
Let  ${\cal R}$  be a dynamic ring    of size $n>3$. 
Let the agents know that the black hole is located in one of three consecutive nodes 
 $H=\{v_1,v_2,v_3\}$ (different from the home-base).
It  is impossible for two colocated  agents to locate the black hole and terminate.
The impossibility holds even if the nodes are equipped with whiteboards, the agents  have distinct  IDs, and the ring is oriented.  
 \end{lemma}
\begin{proof}
Let $a$ and $b$ be the two agents. 
By contradiction, let ${\cal A}$ be  an algorithm that correctly locates the black hole regardless of the pattern of edge  disappearance   in the ring.  
Note that  the two agents cannot visit for the first time a node in $H$ travelling on the same edge at the same round, otherwise the adversary 
would place the black hole in that node killing both.
At least one of the agents must move to visit $H$. Let us assume, w.l.g., that   $a$ is the first to reach $H$ (to  visit   $v_1$) at some round $r$ or that both agents reach $H$ at  round $r$
($a$ visiting $v_1$ and $b$ visiting $v_3$ from the other side).
 At this point  the adversary, regardless of the position of $b$,  removes edge $e=(v_0,v_1)$. 
Note that, while the edge is missing, the two agents cannot  communicate   because they are disconnected on one side by the missing link and on the other by the black hole. 
Should agent $a$ survive, it has no choice but  visiting $v_2$ (from $v_1$) to determine whether the black hole is in $v_2$ or in $v_3$.
Agent $b$ cannot wait for ever in $v_0$ because $e$ might be permanently missing (and leading to the black hole), and it has to reach $H$ from the other side.  
Once    $b$  reaches $v_3$, if it does not die  it cannot avoid  visiting $v_2$  to determine whether the black hole is in $v_1$ or in $v_2$. Hence, within finite time they would both enter $v_2$, albeit at different moments. By choosing the black hole to be  in $v_2$, both agents  die,    contradicting the correctness of the algorithm.  
The adversary can now reactivate    edge $e$.

 \end{proof}
 From the above technical lemma is immediate that:

\begin{theorem}\label{ic2impossibile}
In a dynamic ring    of size $n>3$,
two colocated agents  cannot solve the \BHS.
%
 The impossibility holds even  if    the agents have unique IDs, 
 and   are equipped with the strongest ({\em Whiteboard}) communication  model.
\end{theorem}

Interestingly, we can show that there is no algorithm solving \BHS\ if $n$ is unknown. Such result does not depend on the number
of agents. 
 \begin{theorem}\label{icnoNimpossible}
There exists no algorithm that solves the \BHS\ in a dynamic ring  ${\cal R}$  whose size is unknown to the agents.
The result holds even if  the nodes have   whiteboards, the agents have IDs, and  irrespectively of the number of agents.
\end{theorem}
\begin{proof}  
The proof is by contradiction. Let ${\cal A}$ be a correct algorithm, and let the adversary remove an   edge $e$   at round  $0$. 
For the algorithm to be correct, there must exist a round $r$ when the  the portion of the ring delimited by the black hole and by edge $e$ is fully explored  (with at least one agent dead in the black hole and the corresponding link marked as dangerous). Not knowing the size of the ring, as long as $e$ is missing,  the remaining agents cannot decide whether they have explored all the nodes of the ring and can terminate, or whether the ring is larger, the missing edge $e$ is not incident to the black hole,  and there is still a portion to explore. If they decide to terminate, the adversary  will make $e$ re-appear revealing the  unexplored part of a larger ring and
  ${\cal   A}$ would be incorrect; if instead  they decide to wait for $e$ to re-appear,   the missing edge will   be permanently missing, and  the agents will never terminate.
%
 \end{proof}

Next recall the  following obvious fact:

\begin{observation}\label{trivialid:coloc}
Anonymous colocated agents  with an endogenous communication mechanism (i.e., F2F or Vision model) cannot solve \BHS\ in a 
 static ring,  regardless of their number.  
\end{observation}


We now introduce a  technical lemma   
that will be used to establish lower bounds in the rest of the paper.  The lemma is based on the following observation

\begin{observation}[\cite{ErHK15}]\label{madebyus}
Given a dynamic ring ${\cal R}$, and a cut $U$ (with $|U|>1$) of its footprint connected by edges $e_c$ and $e_{cc}$ to nodes in $V \setminus U$. Assume that at round $r$ all the agents are in $U$. 
 If at round $r$ there is not an agent that tries to traverse $e_{c}$ and an agent that tires to traverse $e_{cc}$, then the adversary may prevent agents to visit a node outside $U$.  
\end{observation}

We say that a node $v$ is explored if it is has been visited at least one time by an agent.  Let us assume that agents are colocated and let us use $U_r$ to denote the set of explored nodes at round $r$.
Note that $U_{r}$ must be a cut of the ring, and that $U_{r-1} \subseteq U_r $. 
We will also say that a round $r$ is an {\em expansion round} if $U_{r} \subset U_{r+1}$, and that agent $a$ communicates with another agent $b$ after round $r$, if either $a$ and $b$ meet at a round $r'>r$, or at a round $r'>r$ agent $b$ visits a node on which $a$ wrote something in a round $r''$ with $r < r'' \leq r'$.

\begin{lemma}\label{made:constant}
If ${\cal A}$ solves the \BHS\ with ${\cal O}(n \cdot f(n))$ moves using three agents, then there must exist an agent $a$ that explores a sequence $seq$ of at least ${\Omega}(\frac{n}{f(n)})$ nodes such that: \\
-  $a$ does not communicate with any other agent while exploring nodes in $seq$.\\
- $a$ visits at most $o(n)$ nodes outside $seq$ while exploring nodes in $seq$. \\
The lemma holds even if the agents are colocated, they have distinct IDs, and the nodes are equipped with  whiteboards. 
\end{lemma}
\begin{proof} 
Let us have three agents starting from node $v_0$. It is easy to see that, if ${\cal A}$ is correct, then there is a round $r$ in which $|U_r|={ \Theta}(n)$ and $|V \setminus U_r| =\Theta(n)$: as long as the black hole \BH $\not\in U_{r}$ then ${\cal A}$ cannot correctly terminate, and \BH can be at a clockwise and counter-clockwise distance that is ${ \Theta}(n)$ from $v_0$. 

Let round $r' \geq r$ be an {\em expansion round} in which a new node $v_{x_1}$ is explored. 
By Observation~\ref{madebyus} we have that at each expansion round there are two agents trying to cross the edges of the cut, we name them the pushing agents. Note that the adversary may always allow only one the pushing agent to explore a new node, and block the other. 

Without loss of generality, we name exploring agent, shortened in $e$, to refer to the pushing agent that explores a new node, and as $b$ and $c$ the others.  Note that it may be  possible that the three agents alternate their roles of pushing agents and so that exploring agent is not always the same agent to explore. However in this case we just rename the agents, as we are using the name $e$ just fo convenience since the proof does not depend on $a$ being always the same agent. 

By Observation~\ref{madebyus} and from the fact that the adversary can decide which agent will explore a node outside $U_{r'}$ by orchestrating the removal of the edges, we can assume that in each expansion round the distance between $a$ and the other agents is ${\Theta}(n)$. 

Therefore, in order for $e$ to communicate with $b$ and $c$ after each expansion round $r' > r$, at least ${\Theta}(n)$ moves are needed. 

Now let $v_{x_1},v_{x_2},\ldots, v_{x_{t}}$ be a sequence of $t$ consecutive nodes that $e$ explores from round $r'$ to round $r''$, such that (i) $e$ never communicates with other agents from round $r'$ to round $r''$, and (ii) $a$ visits at most $o(n)$ nodes after each exploration of a node in the sequence. 
Let us call such a sequence a {\em solitary sequence} of length $t$. Note that when $e$ explores a single node and immediately communicates with others, we have a solitary sequence of length $1$. Therefore, after round $r$, each time we explore a new node, we create a solitary sequence of length at least $1$. 

We now argue that after round $r$ at most $s={\cal O}(f(n))$ solitary sequences are generated. Suppose the contrary, by definition when a solitary sequence ends ${\Omega}(n)$ moves are executed, and, by hypothesis ${\cal A}$ runs in ${\cal O}(n \cdot f(n))$ moves. If we had ${\omega}(f(n))$ solitary sequences we would have ${\omega}(n \cdot f(n))$ moves  (recall that $f(n) = \omega (g(n))$ if $\lim\limits_{n \to +\infty} \frac{f(n)}{g(n)}=+\infty$). 
 However, there are still $k={\Theta}(n)$ nodes to explore after round $r$, and this has to be done with $s={\cal O}(f(n))$ solitary sequences. Therefore, there must exists at least one solitary sequence of length at least $\frac{k}{s}={\Omega}(\frac{n}{f(n)})$. Such a solitary sequence proves our claim.
\end{proof}

Intuitively, Lemma \ref{made:constant} says that, in any \BHS algorithm that has cost (i.e., number of moves) ${\cal O}(n \cdot f(n))$, there exists at least one agent that explores a sequence of nodes of length ${\Omega}(\frac{n}{f(n)})$;
during the exploration of such sequence the agent does not communicate with others (either by a direct meeting, or by writing on a whiteboard or leaving a pebble on node that is visited by others), 
and it visits at most $o(n)$ nodes outside the ones in the sequence.

\section{Preliminaries}\label{sec:preliminaries}

Before presenting and analyzing our solution protocols, we briefly describe a well known 
idea employed for \BHS\ in static graphs that will be adapted in our algorithms, as well as
the conventions and symbols used in our protocols.

\subsection{Cautious Walk}

{\em Cautious Walk} is  a mechanism introduced in \cite{Dobrev2002}  for agents to move on dangerous graphs in such a way that two (or more) agents never enter the black hole from the same edge. The general idea of cautious walk in static graphs is that when an agent $a$ moves  from $u$ to $v$ through an unexplored (thus dangerous) edge $(u,v)$, $a$  must leave  the information that the edge is under exploration at $u$. The information can be provided through some form of mark in case of exogenous communication mechanisms, or implicit   in case of endogenous  mechanisms  (e.g., by employing a second agent as a ``witness''). 
In our algorithms we will make use of variants of the general idea of cautious walk, adapting it  to  the dynamic case and to the particular model under discussion.


\subsection{Pseudocode Convention}\label{sec:conv}
We use the pseudocode convention introduced in \cite{DiLDFS20}. 
In particular, our algorithms use  as a building block  procedure \Explore($dir$ $|$ $p_1:s_1$; $p_2:s_2$; \dots; $p_k:s_k$),
where $dir \in\{ \dLeft, \dRight, \dNil\}$, $p_i$ is a predicate, and $s_i$ is a state. 
In Procedure \Explore, the agent takes a snapshot, then evaluates  predicates $p_1, \ldots, p_k$ in order;
as soon as a predicate is satisfied, say $p_i$, the procedure exits, and the agent transitions to the state $s_i$ specified by $p_i$. 
If no predicate is satisfied, the agent tries to move in the specified direction $dir$ (or it stays still if $dir=\dNil$), and the procedure is executed again in the next round. 
The following are the main predicates  used in our Algorithms:
\begin{itemize}
\item \pMeeting[ID]: the agent sees another agent with identifier $ID$ arriving at the node where it resides, or the agent arrives in a node, and it sees another agent with identifier $ID$. 
\item \pSees[ID]: the agent sees another agent with identified $ID$ in the node where it resides.
\end{itemize}
Furthermore, the following variables are maintained by the algorithms:
\begin{itemize}
\item \vTtime, \vTsteps: the total number of rounds and distinct visited nodes, respectively, since the beginning of the execution of the algorithm.

\item \vEtime, \vEsteps: the total number of rounds and distinct visited nodes, respectively, since the last call of procedure \Explore.
\item \vMtime[C/ (CC)]: the number of rounds during which the clockwise/ (resp. counter-clockwise) edge is missing since the last call of procedure \Explore.
\item $\#Meets[ID]$: the number of times the agent has met with agent $ID$.
\item $RLastMet[ID]$  records the number of rounds elapsed since the agent has seen (or meet) an agent with id $ID$
\end{itemize}

Observe that, in a fully synchronous system, when predicate \pMeeting[y] holds for an agent $a$ with id $x$, then predicate \pMeeting[x] holds for the agent with id $y$.  

\section{Colocated Agents}\label{sec:colocated}

In this section we study the \BHS\ in interval connected rings when agents are {\em colocated}, i.e., they start from the same home-base.  In the following we distinguish between endogenous and exogenous communication mechanisms.


\subsection{Endogenous Communication Mechanisms}\label{f2f:com}

First  recall that, with endogenous communication mechanisms, 
  IDs are necessary for \BHS (see Observation~\ref{trivialid:coloc}).   Hence, we assume the agents have unique IDs.
%

\subsubsection{Lower bound on Cost and Time}\label{sec:f2fbound}

In this section we present a quadratic lower bound on the number of moves and on the number of rounds needed by any algorithm of optimal size to solve   \BHS.

\begin{theorem}\label{th:icf2fbound}
Given a dynamic ring ${\cal R}$, any algorithm ${\cal A}$ that solves the \BHS\ with three agents and an endogenous communication mechanism 
has a cost of at least  ${\Omega}(n^2)$ moves and needs ${\Omega}(n^2)$ rounds.
The result holds even if  the agents are colocated,  have distinct IDs, and  the  model is F2F. 
\end{theorem}
\begin{proof}
Let $a, b, c$  be the three agents.  We first show the bound on the number of moves. 
The proof proceeds by contradiction: let ${\cal A}$ be a solution algorithm performing ${o}(n^2)$ moves.  
By Lemma \ref{made:constant}, when $n$ is large enough, in ${\cal A}$ there exists a round $r$ such that, by the end of round $r$, agent $a$ explored at least three nodes, say $v_1$, $v_2$ and $v_3$, without communicating with $b$ and $c$. 
Let the black hole be  one of these three nodes; hence, by round $r$,  agent $a$ is eliminated. 
At this point, the two remaining agents, $b$ and $c$, even if aware of $a$'s demise, are unaware 
of which of $v_1$, $v_2$ and $v_3$ is the \BH. By  Lemma~\ref{lm:technical}, the agents cannot determine the exact location and, hence,
${\cal A}$ cannot correctly solve the problem;  a contradiction. 

For the bound on the rounds, as just shown, 
 ${\cal A}$  must  perform ${\Omega}(n^2)$ moves. Even if these moves were equally divided among the three agents and performed in parallel, we would have at most three moves at each round. Therefore, the number of rounds is quadratic, and the claim follows. 
\end{proof}


\subsubsection{An Optimal Solution: \RT}\label{algorithm:rt}

 From Theorem \ref{ic2impossibile} we know that the minimum size of any solution algorithm is three.  
 From Theorems~\ref{th:icf2fbound},   we know that cost and time are    $\Omega(n^2)$. In this Section, we show that the bounds are tight describing a size  optimal solution  with $O(n^2)$ moves and rounds.  Our solution works in the weakest   communication model (Vision).
 
   \begin{algorithm*}

\begin{algorithmic}
\State States: \{{\sf Init}, {\sf NewNode}, {\sf Return}, {\sf Move}\}.
\AtState{Init, NewNode}
    \State \Call{Explore}{\dRight$|$ $\vEsteps>0$: \sReturn}
\AtState{Return}
     \State \Call{Explore}{\dLeft$|$ $\vEsteps>0$: {\sf Move}}
     \AtState{Move}
     \State \Call{Explore}{\dRight$|$ $\vEsteps>0$: {\sf NewNode}} 
     
\END
\end{algorithmic}
\caption{Algorithm \RT for \A \label{alg-3a}}
\end{algorithm*}
 
   \begin{algorithm*}

\begin{algorithmic}
\State States: \{{\sf Init}, {\sf Bounce}, {\sf Return}\}.
\AtState{Init}
\State $\mathit{nextTarget} \gets 1$
    \State \Call{Explore}{\dLeft$|$ $\vEsteps \geq \mathit{nextTarget} $: \sReturn}
\AtState{Return}

     \State \Call{Explore}{\dRight$|$ \pSees[\Leader]: \sBounce}
     \AtState{Bounce}
         \State $\mathit{nextTarget} \gets \vEsteps+1$
      \State \Call{Explore}{\dLeft$|$ $\vEsteps \geq nextTarget$: \sReturn}
\END
\end{algorithmic}
\caption{Algorithm \RT for \R \label{alg-3r}}
\end{algorithm*}
 
Our \RT algorithm works using three agents: the {\sc Avanguard}, the \R, and  the \Leader (refer to Figure \ref{fig:cp}, which shows two examples of possible runs of the algorithm). 
Initially, all three agents are on the same node $v_0$, the home-base.

\A and \Leader move clockwise ``cautiously'': 
If   the edge $e$ in the clockwise direction is not present, both \Leader and \A wait until it  reappears.

If edge $e$ is present, \A moves to the unexplored node using edge $e$. Then, 
if in a successive round the edge $e$ is present, \A goes back to \Leader,  signalling 
that the recently visited node is safe;  at this point,  both \Leader and \A safely move clockwise to the recently explored node
If  \A does not return when $e$ is present,
 \Leader knows the position of the black hole (the node just visited by \A) and terminates;
 in this case, we say that \A{} {\em fails to report}. 

While \A and \Leader are performing this special exploration, \R moves as follows:
it goes counter-clockwise until it visits the first unexplored node; then, it goes back clockwise until it meets again \Leader.
Once \R meets \Leader, it reverts back its movement direction to counter-clockwise, iterating the same kind of move: in other words, it swings similarly to a pendulum that increases its counter-clockwise amplitude of one node 
at each oscillation. 

In case \R finds a missing edge on its path, it waits until the edge re-appears; and then  it keeps performing the oscillating movement. 
We say that \R\space {\em fails to report} to \Leader if the \Leader sees a missing edge in its clockwise direction and, despite waiting for this edge to appear for a time long enough for \R to explore a new node and go back, it does not meet \R. Intuitively, since at most one edge is missing at each round, the fail to report of \R implies that \R entered the black hole. 
Note that, in this case, \Leader can compute exactly the position of the black hole. 

    \begin{algorithm*}

\begin{algorithmic}
\State  Predicates Shorthands: $FailedReport[\A] =\vEtime > \vMtime[C]$.
\State   $FailedReport[\R]= \vMtime[C]> 2 ((\#Meets[\R]+1)+\vTsteps)$.

\State States: \{{\sf Init}, {\sf Cautious}, {\sf Move}, {\sf TerminateA},  {\sf TerminateR}\}.
\AtState{Init,Cautious}

    \State \Call{Explore}{\dNil $|$  $\pMeeting[\A]$: {\sf Move}; 
       $FailedReport[\A]$:  {\sf TerminateA}; \newline $FailedReport[\R]$: {\sf TerminateR} }
%
%
\AtState{Move}

      \State \Call{Explore}{\dRight$|$ $\vEsteps>0$: {\sf Cautios};  $FailedReport[\R]$: {\sf TerminateR}}
     \AtState{TerminateA}
      \State Terminate, \BH is in the next node in clockwise direction. 
           \AtState{TerminateR}
      \State Terminate, \BH is in the node at  counter-clockwise distance $\#Meets[\R]+1$  from $v_0$. 
\END

\end{algorithmic}
\caption{Algorithm \RT for \Leader \label{alg-3l}}
\end{algorithm*}

 \begin{figure}
\center
  \begin{subfigure}[b]{0.45\textwidth}
  \center
    \includegraphics[width=\textwidth]{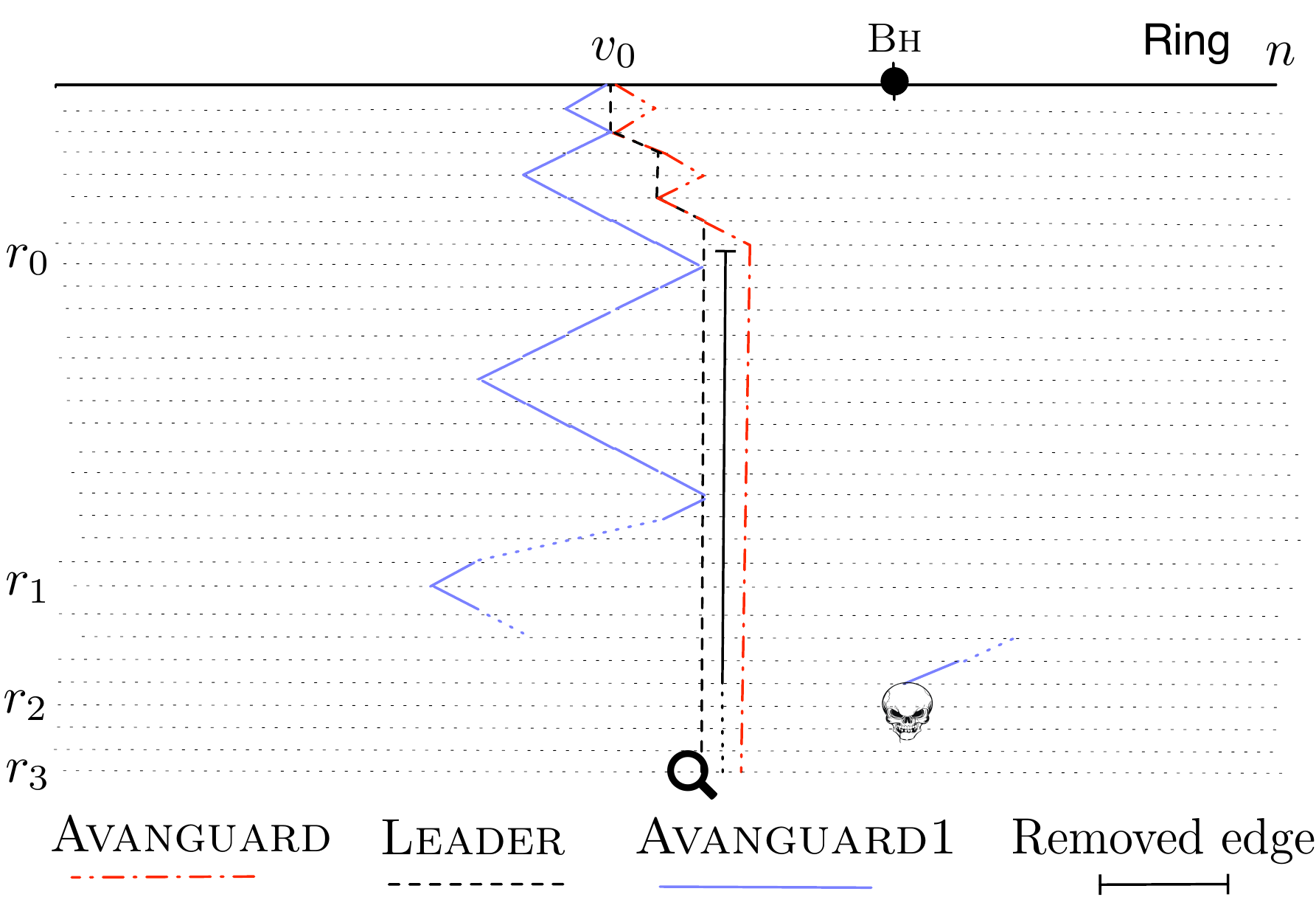}
    \caption{Example of a run where the termination is due to \R. At round $r_0$ \A and \Leader are blocked
    by a missing edge. At round $r_2$ \R enters in the black hole. At round $r_3$ \R fails to report. }
    \label{cp:case1}
  \end{subfigure}
  \,\,
  \begin{subfigure}[b]{0.45\textwidth}
    \includegraphics[width=\textwidth]{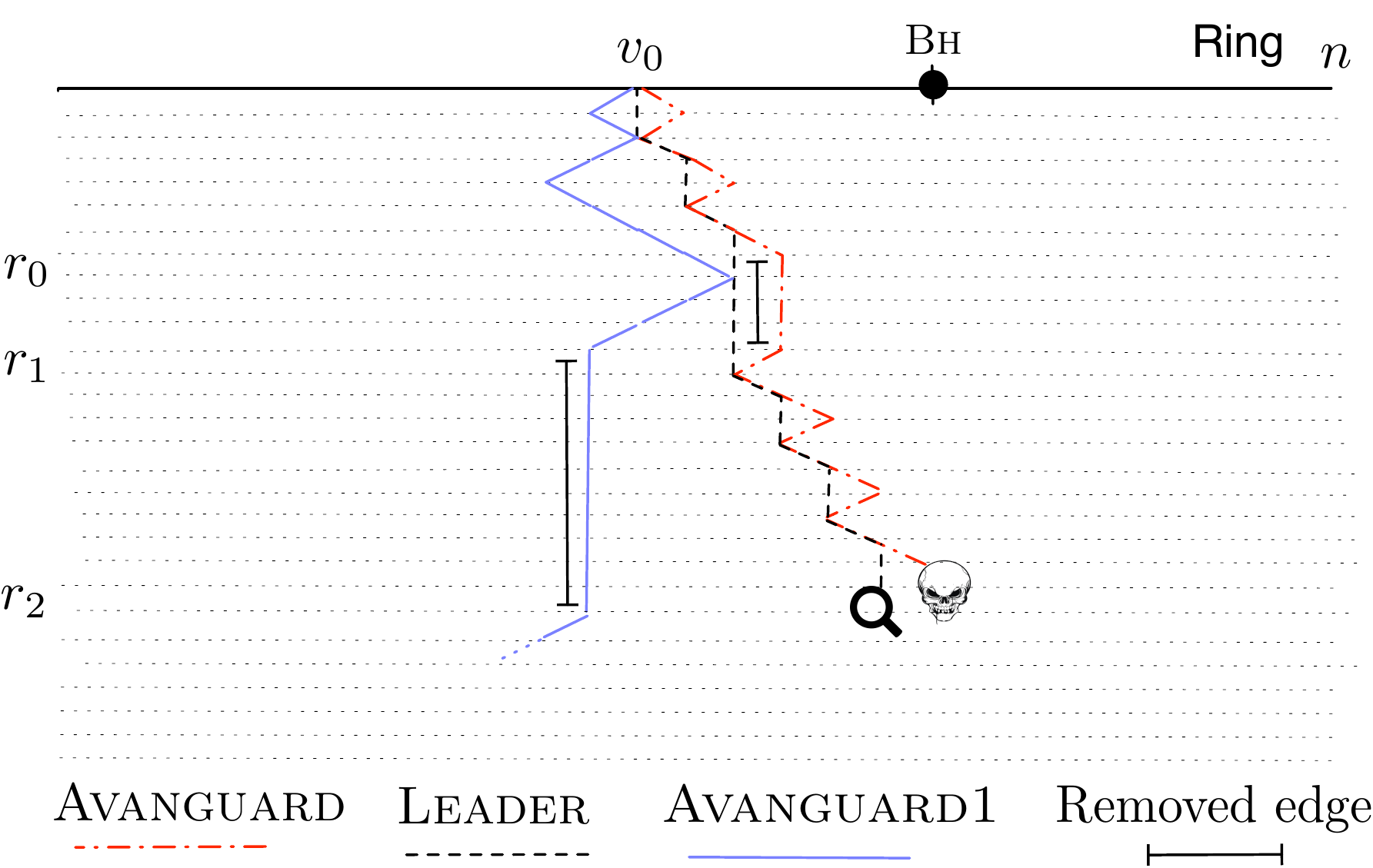}
    \caption{Example of a run where the termination is due to \A.  At round $r_1$ \A and \Leader are free to go while \R is blocked (at most one edge can be removed at each round). 
At round $r_2-1$ \A enters in the black hole. At round $r_2$ \A fails to report.}
    \label{cp:case2}
  \end{subfigure}
  \caption{Example runs of algorithm \RT.  The black hole is the node \BH and it is indicated by a black dot. When no edge is missing agents move as follows. The \R moves in the counter-clockwise direction, while \A and \Leader move
    in clockwise direction. The \A and the \Leader move in a coordinated way: the \A goes forward, then back and at this point both do a step forward (see the
    zig-zag line of \A). Agent \R moves as a pendulum exploring a new node left at each swing.  In Figure~\ref{cp:case1} is depicted a termination due to the failure to report by \R. In Figure~\ref{cp:case2} the termination is due to the failure to report by \A.}\label{fig:cp}
\end{figure}

\paragraph{Correctness}

In this section we prove the correctness of algorithm  \RT 
described in the previous section (refer to the 
 pseudocode of Algorithms~\ref{alg-3a},~\ref{alg-3r}, and~\ref{alg-3l}).
\begin{lemma}\label{lm:rt:beneter}
Consider three agents executing \RT. 
Let $r$ be the first round in which one agent enters the black hole. We have $r={\cal O}(n^2)$.
\end{lemma}
\begin{proof}
Recall that \R moves counter-clockwise of $nextTarget$ steps, and then it moves clockwise until it meets \Leader: the distance it travels during this movement is upper bounded by $2n$. This implies that if \R is not blocked by a missing edge, it explores a new node every $2n$ rounds at most. 
Since at most one edge is missing at each round, and since by construction \A moves clockwise and \R moves counter-clockwise exploring disjoint portion of the rings, they can never be both blocked at the same round.
Moreover, if one of them is blocked for $2n$ rounds (even not consecutively), the other explores at least a node. 
It follows that at least
one new node is explored by one of these two agents every (at most) ${\cal O}(n)$ rounds. 
The above implies that, in at most ${\cal O}(n^2)$ rounds, \R or \A reaches the black hole. 
\end{proof}
\begin{observation}\label{ob:rt:leadernodies}
Consider three agents executing \RT. 
The leader never enters in the \BH.
\end{observation}
\begin{proof}
The \Leader moves into a node only after \A explored it and went back to signal the node as safe: therefore, \Leader can never enter the black hole. 
\end{proof}
\begin{lemma}\label{lm:rt:term}
Consider three agents executing \RT.  Let $r$ be the round in which one agent enters  \BH. Then  \Leader terminates by a round $r_f=r+{\cal O}(n^2)$.\end{lemma}
\begin{proof}
By Observation~\ref{ob:rt:leadernodies}, at round $r$ \A or \R entered the black hole. 
First  observe that for the \Leader we always have $\#Meets[\R] \leq n$ and $\vTsteps \leq n$. The bound  $\vTsteps \leq n$ follows from the fact that the \Leader cannot do more than $n$ steps otherwise it would enter the black hole.
The bound $\#Meets[\R] \leq n$ derives from the fact that \R explores a new node (in counter-clockwise  direction) for each meeting with \Leader; thus, after  at most $n$ meetings, \R entered the black hole. 
We distinguish the two possible cases:

\noindent 
$1)$ \A reaches the black hole at round $r$. By construction, at round $r$ the \Leader is in state {\sf Cautious } in the counter-clockwise neighbor of the black hole. If  edge $e$ between the \Leader and the black hole is present at a round $r'>r$, then \Leader will not see \A returning at round $r'+1$; hence, it terminates by predicate $FailedReport[\A]$. Therefore, let us assume that  edge $e$ is missing from round $r$ on. In this case, \R cannot be blocked and will explore a new node every at most $2n$ rounds (recall that \R oscillates). Now, the only scenario that can occur is that also \R enters the black hole (from the counter-clockwise neighbor of the black hole): again, no more than one edge is missing and it has to be $e$, therefore, \R is never blocked. Consequently, in ${\cal O}(n^2)$ rounds \R enters in the black hole. Recall, that at this moment on the \Leader we have $\#Meets[\R] \leq n$ and $\vTsteps \leq n$.
		Therefore, after  ${\cal O}(n)$ additional rounds \R will fail to report to \Leader, and the predicate $FailedReport[\R]=\vMtime[C]> 2n+2 > 2*((\#Meets[\R]+1)+\vTsteps) $ will be verified.
		
\noindent 
$2)$ \R reaches the black hole at round $r$. 
Let us suppose the \Leader is at a node $v$ whose clockwise edge is missing; if this edge is missing for more than $2 \cdot n$ rounds, then the failure to report of \R  triggers, and the \Leader terminates. 
Therefore, let us assume that no edge in the clockwise direction of movement of \Leader is ever missing for more than $2 \cdot n$ rounds. In this case, \Leader and  \A will eventually reach the black hole from the clockwise direction, in at most in ${\cal O}(n^2)$ rounds. Once \A enters the black hole, \Leader terminates in at most $2 \cdot n$ rounds, either by the failure to report check by \R (see predicate $FailedReport[\R]$) or by \A not returning to \Leader (predicate $FailedReport[\A]$).  
\end{proof}

\begin{lemma}\label{lm:rt:corrterm}
Consider three agents executing \RT.  If the \Leader terminates then it correctly locates the \BH. 
\end{lemma}
\begin{proof}
We first discuss the relationship between variables $\#Meets[\R]$ and $\vTsteps$ on \Leader and the behaviour of \R. 
\R  moves counter-clockwise of $nextTarget$ steps, and then it moves clockwise until it meets \Leader. 
When \R meets with the leader its variable \vEsteps contains the number of edges traversed from the last explored node, in counter-clockwise direction, and the \Leader, at this point \R
updates its $nextTarget$ as $nextTarget=\vEsteps+1$, and the \Leader updates $\#Meets[\R]=\#Meets[\R]+1$. It is thus clear that the next node \R will explore is at distance $\#Meets[\R]+1$ from $v_0$ in the clockwise orientation ($nextTarget$ starts from $1$). 
It is also clear that, considering the variables $\#Meets[\R]$ and $\vTsteps$ on \Leader, the quantity  $2*((\#Meets[\R]+1)+\vTsteps)$ is an upper bound
on the number of edges that \R traverses to explore a new node and go back to \Leader.

\noindent We can now prove our claim by a cases analysis on the possible termination cases:

\noindent 
- \Leader terminates in state {\sf TerminateA} on node $v$ at round $r$: in this case, \\$FailedReport[\A]$$ =\vEtime > \vMtime[C]$ is verified when \Leader is in state {\sf Cautious}. That is, there has been at least one round
in which the clockwise edge incident to $v$ was present. Let the round $r'$. By construction, in a round $r''<r'$ \A moved in the clockwise direction (recall that \Leader is in state {\sf Cautious}); but at round $r'$ the edge is not missing and agent \A did not return. It follows that \A moved into the black hole, hence \Leader can correctly compute its position.

\noindent 
- \Leader terminates in state {\sf TerminateR} on node $v$ at round $r$:  in this case, \\$FailedReport[\R]$$=vMtime[C] > 2*((\#Meets[\R]+1)+\vTsteps)$ is verified. 
That is, the clockwise edge  incident to $v$ has been missing for at least $2*((\#Meets[\R]+1)+\vTsteps)$ rounds. Since at most one edge can be missing at each round, this interval of time is sufficient for \R to reach an unexplored node at distance $\#Meets[\R]+1$ from $v_0$ in counter-clockwise orientation, and then go back to \Leader (recall that \Leader is at distance \vTsteps from $v_0$ in the clockwise orientation). This implies that \R entered the black hole, at distance $\#Meets[\R]+1$ from $v_0$ in the counter-clockwise direction. Hence, also in this case, \Leader can correctly compute the position of the black hole.
 
\end{proof}

\begin{theorem}\label{th:pendulum}
  Consider  a dynamic ring  ${\cal R}$, with three colocated agents with distinct IDs in the Vision model. 
  Algorithm \RT solves  \BHS\  with ${\cal O}(n^2)$ moves and in ${\cal O}(n^2)$ rounds. 
\end{theorem}
\begin{proof}
By Lemma~\ref{lm:rt:beneter} and Observation~\ref{ob:rt:leadernodies} we have that at a round $r={\cal O}(n^2)$ \R or \A entered the black hole while the \Leader remains alive. 
By Lemma~\ref{lm:rt:term} we have that at a round $r_f=r+{\cal O}(n^2)$ the \Leader terminates, and by Lemma~\ref{lm:rt:corrterm} it correctly locates the \BH. The fact that no other agent terminates incorrectly it is
immediate from the observation that \R and \A have no terminating state. \end{proof}

Summarizing, by Th.~\ref{th:icf2fbound} and Th.~\ref{th:pendulum} we have: 
\begin{theorem}\label{ic3optimal}
Algorithm \RT is size-optimal with optimal cost and time.
\end{theorem}
\color{black}

\subsection{Exogenous Communication Mechanisms}\label{whiteworld}

Note that, with the exogenous  communication mechanisms,  anonymous agents on the same node can easily break the symmetry  and assume different Identifiers   by exploiting the mutual exclusion nature of pebbles and whiteboards.
 
As for the ability of agents to interact, we can observe that, even in the  simpler pebble  model, any communication between agents located at the same node is easy to achieve (e.g.,  two agents may exchange messages of any size   using a communication protocol in which they send one bit every constant number of rounds). 
Therefore, in this section we can assume that the agents are able to communicate. 
 Specifically, the communication of constant size messages is assumed to be instantaneous, since it can implemented trivially by a multiplexing mechanism (the logical rounds are divided in a constant number of physical round, the first of which is used to execute the actual algorithm and the others to communicate).


\subsubsection{Lower bound on Cost and Time}\label{lowerbound:sqrtsn}

 The lower bound of Theorem~\ref{th:icf2fbound} does not hold when employing one of the proposed exogenous communication mechanisms. 
 We now show  lower bounds of  $\Omega(n^{1.5})$ on cost and time complexity; the lower bounds hold even employing the strongest of the exogenous mechanisms: whiteboards.

%

\begin{theorem}\label{th:sqrtn}
Given a dynamic ring ${\cal R}$ in the Whiteboard model, any algorithm ${\cal A}$ solving \BHS\ with three agents requires $\Omega(n^{1.5})$ moves and $\Omega(n^{1.5})$ rounds even
if the agents are colocated and  have distinct IDs.
\end{theorem}
\begin{proof}

Let $a,b$ and $c$ be the tree agents. 
By Lemma~\ref{made:constant}, if ${\cal A}$ terminates in ${o}(n \cdot \sqrt{n})$ moves, then there is an agent, say $a$, that explores a sequence $S=v_1,v_2,\ldots,v_k$ of ${ \Omega}(\sqrt{n})$ consecutive nodes while it does not communicate with any other agent. Let us suppose that the black hole is in $S$. 
Recall that $a$ communicated with the other agents only before starting the exploration of $S$. Hence, when $a$ reaches \BH, neither $b$ nor $c$ know the exact location of the node $a$ was visiting when it entered the black hole. In the best case, when $a$ disappears, agents $b$ and $c$ know that \BH is in   sector $S$, but they do not know the exact location of \BH in $S$ (refer to Figure~\ref{fig:nsqrtn}). 
\begin{figure}

 \centering
  \includegraphics[width=0.5\textwidth]{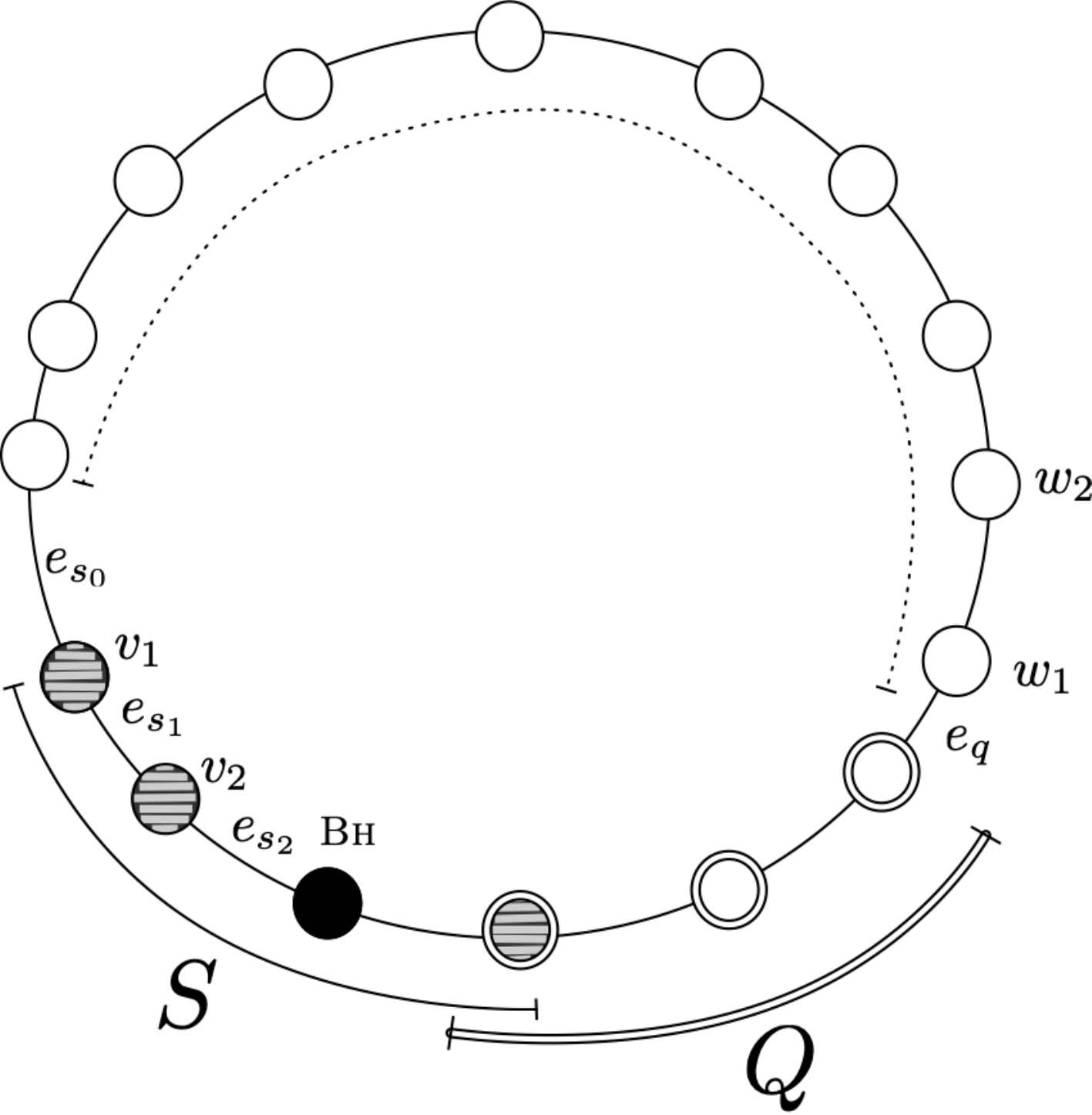}
   \caption{Pictorial representation for the lower bound of $n \sqrt{n}$: the zone $S$ is a set of contiguous nodes where the black hole could be located, the zone $Q$ represents all nodes where $a$ wrote information about the location of $s$. \label{fig:nsqrtn}}
\end{figure}

Since, by hypothesis, nodes are equipped with whiteboards, it is possible that $a$ wrote on a set of nodes $Q$ the exact location of \BH. However, by Lemma~\ref{made:constant}, when $a$ reaches the black hole \BH, neither $b$ nor $c$ visited one of the nodes in $Q$. Without loss of generality, let us assume that $Q$ is a set of contiguous nodes placed at the counter-clockwise side of the \BH (see Figure \ref{fig:nsqrtn}). By Lemma~\ref{made:constant}, $|Q|={o}(n)$.

Also, let $C$ be the cut of the ring that includes both $b$ and $c$ (i.e., $C \cap Q= \emptyset$). By Observation~\ref{madebyus}, if $b$ and $c$ want to visit a node outside $C$, then one of them has to try to traverse the clockwise edge connecting $C$ with other nodes in $V$, while the other agent has to do the same in the counter-clockwise direction. Algorithm ${\cal A}$ may identify the black hole only in two possible ways: either one agent visits a node in $Q$; or an agent, say $b$, reaches \BH while $a$ knows the next node $b$ is just about to visit.

Let $e_{q}$ be the edge connecting a node in $Q$ to a node in $V \setminus S$, and let $e_{s_0}$ be the edge connecting a node in $S$ to a node in $V \setminus Q$ (see Figure \ref{fig:nsqrtn}). We now establish the following strategy for an adversarial scheduler that decides which edge is missing: 
the adversary removes the edges such that no agent will ever traverse $e_q$ (this can be achieved by always removing $e_q$ when an agent is present on a node with $e_q$ incident). Moreover, the adversary never let agents continue their exploration on $S$ as long as one of them is not trying to traverse $e_q$. 
As an example as long as $v_1$ is never explored the adversary always remove $e_{s_0}$ but when an agent may traverse $e_q$, when $v_1$ is explored the same behaviour will be applied to $e_{s_1}=(v_1,v_2)$, and so on.
Thanks to this strategy, no agent will ever learn the location of \BH by visiting a node in $Q$.

Therefore, an agent is forced to reach the black hole by using a counter-clockwise edge ($e_{s_2}$ in Figure~\ref{fig:nsqrtn}). Let us assume that $b$ is the agent that traverses for the first time an edge $e_{s_j}=(v_j,v_{j+1})$ with $v_j \in S$.
Note that it might be possible that is not always $b$ to do that, but each time agent cross each other we may simply assume they swap also identity. 

The strategy of the scheduler forces $a$ to try to traverse $e_{q}$ each time $b$ traverses a new $e_{s_j}$. Also, $b$ needs to communicate with $a$ after every newly explored node in $S$: otherwise, if $b$ explores two nodes without communicating with $a$ and one is the black hole, then $a$ cannot disambiguate which one is \BH (the adversary can perpetually block $a$ on two neighbor nodes by a careful removal of edges). 

Finally, $n$ moves are necessary for $b$ to be able to communicate with $a$: in fact, $a$ is trying to traverse $e_q$ when $b$ traverses an $e_{s_j}$. Therefore, $n|S|$ moves are required to identify a black hole in $S$. The bound now follows immediately from the fact that $|S|={\Omega}(\sqrt{n})$.  The bound on the number of rounds follow immediately from the fact that a constant number of agents needs ${ \Omega}(\sqrt{n})$ rounds to perform ${ \Omega}(\sqrt{n})$ moves.

 \end{proof}

 \subsubsection{An Optimal Solution: \DP }\label{dpalgorithm}
 
  We now describe the \DP, that is  an optimal solution in the weakest of the two exogenous communication mechanisms (i.e., the Pebble model).
  
 \paragraph{Primitives and Pseudocode Conventions in the Pebble model} \label{par:pebble} In our algorithm  we use the term {\em mark} to and {\em unmark} to indicate that an agent is leaving/removing the pebble from a node. 
Besides the \Explore procedure described in Section~\ref{sec:preliminaries}, we also use   procedure \CExplore. 
\CExplore  makes an agent perform  a cautious walk using the pebbles:
the agent marks a node,  moves in the $dir$ direction, then it goes back, unmarks the node, and move again in the $dir$ direction. 
Moreover, to verify  predicate  $\pMeeting[x]$ in this context,  the concerned agents should not be   trying to unmark a node (otherwise they are not considered to have met). 

Note that
this can be implemented by sending messages as discussed at the beginning of Section~\ref{whiteworld}. 
Finally, we  define   predicate ${marked}$  to be verified when the agent resides in a marked node.

\paragraph{High level description}

The algorithm structure is reminiscent of a pendulum that oscillates with two different amplitudes.
At the start \R oscillates by increasing its counter-clockwise amplitude of $\sqrt{n}$ nodes at times:
in other words, \R explores at each oscillation a sector of $\sqrt{n}$ contiguous nodes before reporting to \Leader. While exploring a sector, \R uses the pebble to perform a cautious walk. 
As long as \R is alive, the \Leader and the \A act as in \RT.

The \Leader detects when, and if, \R reached the black hole by using a timeout strategy. Such a timeout strategy is designed in order to not trigger in case \R is just blocked by an edge (see the detailed description). 

 In case \R entered the black hole, the \Leader knows the sector that \R was exploring (\Leader keeps track of the sector by counting the number of times it met \R). Since this sector contains the black hole, it will be denoted as {\em dangerous sector}. 

When (and if) \R fails to report, the \Leader starts moving counter-clockwise looking for the last node marked by \R, while 
\A starts exploring the dangerous sector in the clockwise direction. 

The exploration of the sector by \A is done in a cautious way, by reporting back to \Leader for each newly explored node in this sector 
 (at each swing, the amplitude of \A increases of only $1$ node); thanks to this strategy, \Leader knows which node \A is exploring. Note that, if \A is not blocked by a missing edge it enters in the black hole in ${\cal O}(n \sqrt{n})$ rounds. 
 
At the end, either \Leader reaches the pebble left by \R (and thus, it knows where the black hole is), or \A will
enter the black hole and it fails to report to \Leader (also in this case \Leader can correctly compute the position of the black hole).

  \begin{figure}
\center
  \begin{subfigure}[b]{\textwidth}
  \center
    \includegraphics[width=0.65\textwidth]{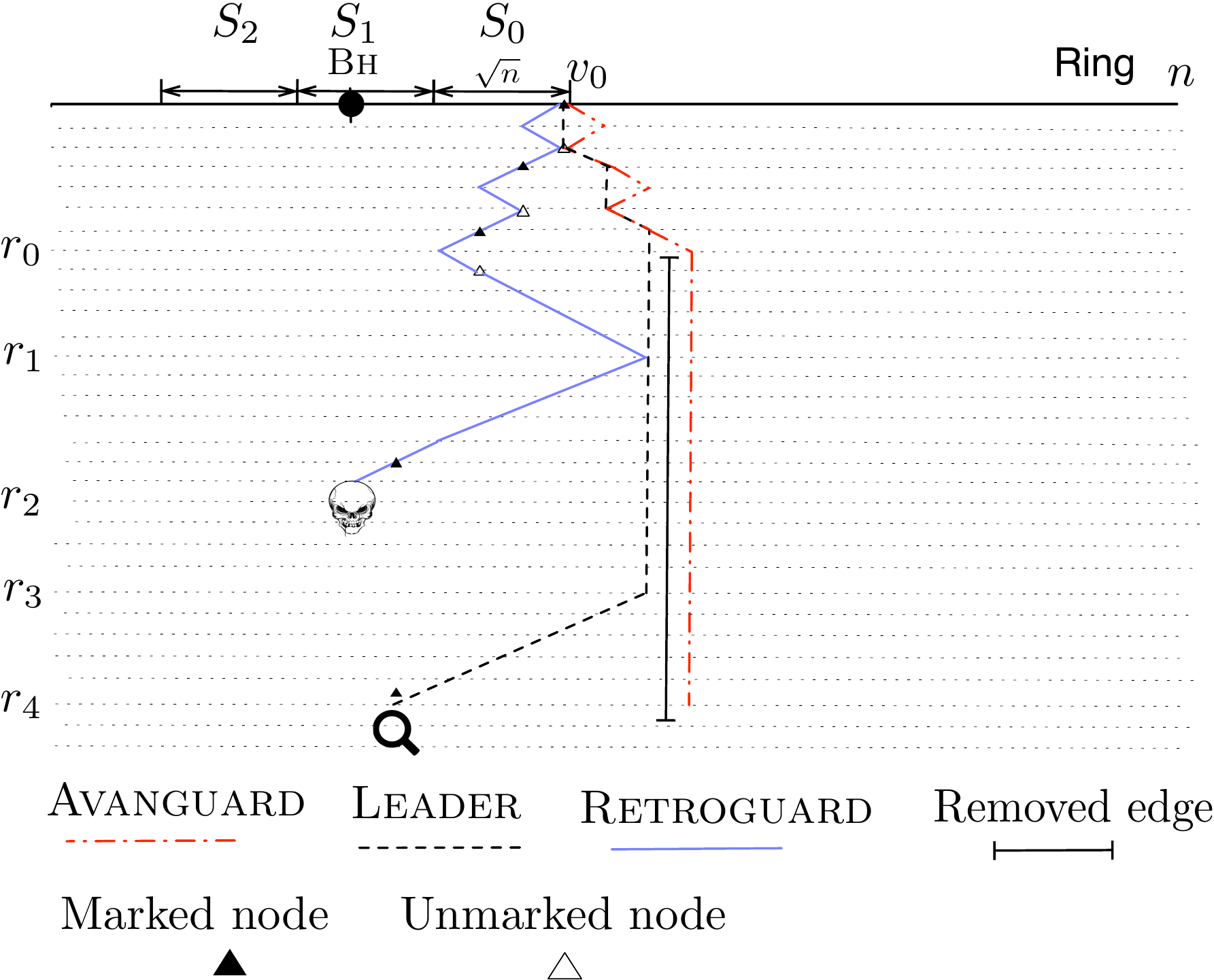}
    \caption{Example of a run where the termination is due to \Leader finding the marked node.  At round $r_0$ \A and \Leader are blocked
    by a missing edge. In the meanwhile agent \R is exploring a sector of size $\sqrt{n}$: note the marking and unmarking of nodes.  
    At round $r_0$ the \R completes the exploration of the first sector $S_0$, and goes back to the \Leader. After meeting with \Leader, it starts exploring sector $S_1$ (we are not showing the cautious walk until the beginning of $S_1$). 
    At round $r_2$ \R enters in the black hole leaving the counter-clockwise neigbhour marked. 
    At round $r_3$ the \Leader detects the failure to report by \R and it changes direction of movement. At round $r_4$ the \Leader finds the marked
    node and terminates.  }
    \label{dp:case1}
  \end{subfigure}
  \begin{subfigure}[b]{\textwidth}
  \center
    \includegraphics[width=0.65\textwidth]{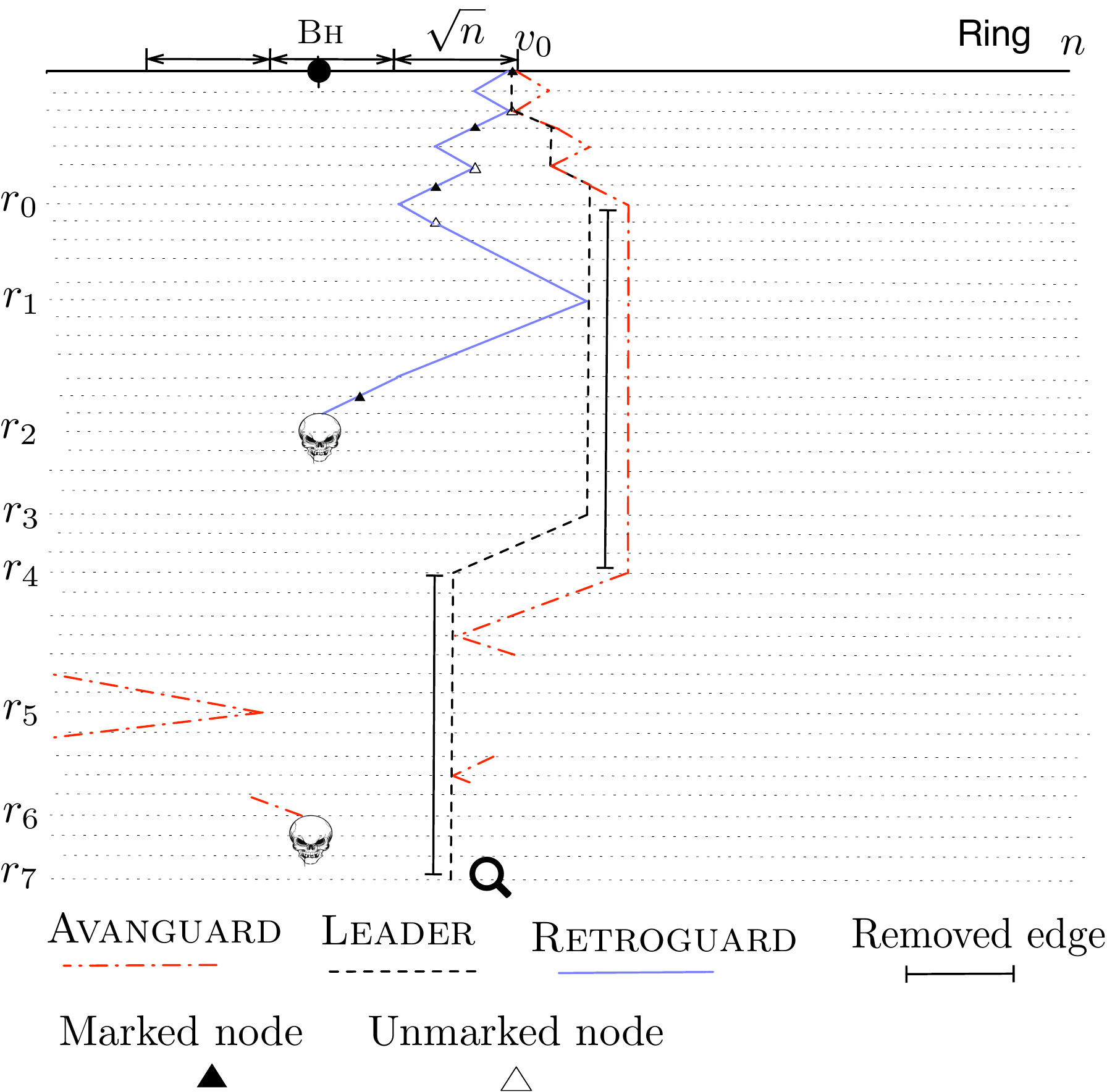}
    \caption{Example of a run where the termination is due to \A. Until round $r_3$ the execution is as the one in Figure~\ref{dp:case1}. 
At round $r_4$ the \Leader is blocked. At the same round \A goes back and it does not find the \Leader, thus it understands that the leader is looking for the node marked by \R. At this point \A goes counter-clockwise until it meets the \Leader,
and it learns the sector that \R was exploring. Once this is done, \A starts exploring the dangerous sector in the clockwise direction: at round $r_5$ it explores the first node of such sector and
goes back to \Leader. At round $r_6$ it explores the second node entering in \BH. At round $r_7$ the \Leader detects the failure to report by \A and terminates.}
    \label{dp:case2}
  \end{subfigure}
  \caption{Example executions of Algorithm \DP. }\label{fig:dp1}
\end{figure}

 \paragraph{Detailed description}
 The pseudocode for \DP is in Algorithms~\ref{alg-dpr},~\ref{alg-dpa}, and~\ref{alg-dpl}. Also, in Figure~\ref{fig:dp1}, two executions are reported. 
 
 As long as \R is not detected dead, the behavior of \Leader and \A is the same as in \RT. 
\R explores counter-clockwise, in a cautious way, sectors of size $\sqrt{n}$ (states {\sf Init} and {\sf Bounce} of Algorithm~\ref{dp:retroguard}).

Once \R explored a sector, it reaches \Leader to report, by moving clockwise (state {\sf Return} of Algorithm~\ref{dp:retroguard}). Once the report is over, \R moves counter-clockwise until
it reaches the end of the next unexplored sector (see the update of $Esteps$ in state {\sf Bounce}; see also the example reported in Figure~\ref{dp:case1}). 

 If \R fails to report back to \Leader, then \Leader goes into the {\sf Detection} state by using a timeout strategy: the transition occurs when the number of rounds from the last meet with \R is $7((\#Meets[\R]+1)\sqrt{n}+\vTsteps)$.
 
 Note that the quantity $(\#Meets[\R]+1)\sqrt{n}+\vTsteps$ is an upper bound on the maximum distance from \Leader and \R: the first component $(\#Meets[\R]+1)\sqrt{n}$ is the counter-clockwise distance between the last node of the sector under exploration and $v_0$, the second quantity $\vTsteps$ is the clockwise distance between the actual position of \Leader and $v_0$ (in the proof will become apparent why we need to multiply this quantity by the constant $7$). We remark that the factor $\vTsteps$ avoid that the \Leader timeous if \R is blocked on an edge, as long as \R is blocked the \Leader moves increasing $\vTsteps$ and delaying the timeout.

Once in state {\sf Detection}, the \Leader walks counter-clockwise trying to reach the last node marked by \R (see state {\sf Detection} of Algorithm~\ref{dp:leader}). In state {\sf Detection} the \Leader also resets $\#Meets[\A]$ to 0, since it is interested in counting the number of times it meets \A from the state switch. A \Leader in {\sf Detection} state terminates if either it finds the marked node, or $3n$ rounds passed without meeting \A (refer again to the example of Figure~\ref{dp:case1}). 

The \A detects that \R disappeared by recognizing that the \Leader moved in a way that is not compatible with the simulated cautious walk. 
More specifically, \A goes into the {\sf SearchLeader} state if in state {\sf NewNode} or state {\sf Move} it does not see the \Leader. 

While in this state, \A moves counter-clockwise until it meets \Leader; when (and if) this occurs,
they both start a communication protocol in which \A reads the variable $\#Meets[\R]$ from \Leader's memory\footnote{Recall that at the beginning of Section~\ref{whiteworld} we discussed how pebbles can be used to communicate messages of non-constant size.}, and it starts oscillating between state {\sf Detection2} and {\sf Return1}. 
The behavior of \A in these two states is as follows: \A goes clockwise until it reaches the first unexplored node in the dangerous sector (the sector that \R was exploring when failed to report); the dangerous sector's position can be computed
using $\#Meets[\R]$: in fact, its position is $\#Meets[\R]\sqrt{n}$ nodes away, according to the counter-clockwise direction from $v_0$.
When a new node in the dangerous sector has been explored, \A moves counter-clockwise until it meets \Leader. Thanks to this mechanism, \Leader always knows the precise location of the next node \A is just about to explore (refer to the example in Figure~\ref{dp:case2}). 

  \begin{algorithm*}

\begin{algorithmic}[1]
\State States: \{{\sf Init}, {\sf Bounce}, {\sf Return}\}.
\AtState{Init}
    \State \Call{CautiousExplore}{\dLeft$|$ $\vEsteps \geq \sqrt{n} $: \sReturn}
\AtState{Return}

     \State \Call{Explore}{\dRight$|$ \pSees[\Leader]: \sBounce}
     \AtState{Bounce}
         \State $\mathit{steps} \gets  \vEsteps+\sqrt{n} $
  \State \Call{CautiousExplore}{\dLeft$|$ $\vEsteps \geq steps $: \sReturn}\END
\end{algorithmic}
\caption{Algorithm \DP for \R \label{alg-dpr}}\label{dp:retroguard}
\end{algorithm*}
 
   \begin{algorithm*}

\begin{algorithmic}[1]
\State States: \{{\sf Init}, {\sf NewNode}, {\sf Return}, {\sf Move}, {\sf SearchLeader}, {\sf Detection1}, {\sf Return1}, {\sf Detection2}.\}
\AtState{Init, NewNode}
    \State \Call{Explore}{\dRight$|$  $\neg\pSees[\Leader]$:{\sf SearchLeader}; $\vEsteps>0$: \sReturn}
\AtState{Return}
     \State \Call{Explore}{\dLeft$|$ $\vEsteps>0$: {\sf Move}}
     \AtState{Move}
     \State \Call{Explore}{\dRight$|$  $\neg\pSees[\Leader]$:{\sf SearchLeader}; $\vEsteps>0$: {\sf NewNode}} 
          \AtState{SearchLeader}
     \State \Call{Explore}{\dLeft$|$ $\pMeeting[Leader]$: {\sf Detection1}} 
          \AtState{Detection1}
          \State Communicate with the \Leader to compute $\mathit{nextTarget}$ that is the distance to first unexplored node in the dangerous sector from clockwise direction.
     \State \Call{Explore}{\dRight$|$ $\vEsteps \geq nextTarget$: {\sf Return1}} 
     \AtState{Return1}
     \State \Call{Explore}{\dLeft$|$ \pMeeting[\Leader]: {\sf Detection2}}
               \AtState{Detection2}
          \State  $\mathit{nextTarget} \gets  \vEsteps+1$
     \State \Call{Explore}{\dRight$|$ $\vEsteps \geq nextTarget$: {\sf Return1}} 
\END
\end{algorithmic}
\caption{Algorithm \DP for \A \label{alg-dpa}}
\end{algorithm*}

 \begin{algorithm*}

\begin{algorithmic}[1]
\State  Variables: $RLastMet[X]=$ number of round since the last meeting of \Leader and agent $X$. 
\State  Predicates Shorthands:
\State $FailedReport[\A] =\vEtime > \vMtime[C]$.
\State $FailedReport[\R]= RLastMet[\R]> 7((\#Meets[\R]+1)\sqrt{n}+\vTsteps)$.
\State $FailedReportD =RLastMet[\A] > 3n$.

\medskip
\State States: \{{\sf Init}, {\sf Cautious}, {\sf Move},  {\sf Detection}, {\sf TerminateA},  {\sf TerminateR},  {\sf TerminateAD} \}

\AtState{Init, Cautious}

    \State \Call{Explore}{\dNil $|$  $\pMeeting[\A]$: {\sf Move}; 
       $FailedReport[\A]$:  {\sf TerminateA};  \newline $FailedReport[\R]$: {\sf Detection} }
\AtState{Move}

      \State \Call{Explore}{\dRight$|$ $\vEsteps>0$: {\sf Cautios};  $FailedReport[\R]$: {\sf Detection}}

           \AtState{Detection}
           \State $\#Meets[\A] \gets 0$ 
           \State $D \gets$ number of nodes explored in the dangerous sector by $\A$
   \State \Call{Explore}{\dLeft$|$ $marked$: {\sf TerminateR};  $FailedReportD$: {\sf TerminateAD} } 
        \AtState{TerminateA}
      \State Terminate, \BH is in the next node in clockwise direction. 
              \AtState{TerminateR}
      \State Terminate, \BH is in the next node in counter-clockwise direction. 
              \AtState{TerminateAD}
      \State Terminate, \BH is the last node \A was exploring. This is computed using $\#Meets[\A]$ and $\#Meets[\R]$.  Specifically, \BH is the node
      that is $(\#Meets[\A]+D)$ nodes way from the clockwise end of the dangeours sector, where $D$ is the number of nodes \A explored in the dangerous sector during the simulation of the cautious walk. 
 \END

\end{algorithmic}
\caption{Algorithm \DP for \Leader \label{alg-dpl}} \label{dp:leader}
\end{algorithm*}

 \paragraph{Correctness.}
 
 \begin{definition}
A {\em sector} is a sequence of ${\sqrt n}$ nodes that \R explores in state {\sf Bounce}. The sequence of sectors that \R explores is denoted by $S_0,S_1,\ldots$,  where $S_i$, $i \geq 0$, is the $i$-th sector explored by \R. 
 \end{definition}

\begin{lemma}\label{gooddec}
If \Leader enters in state {\sf Detection}  at round $r$, then  \R reached the black hole in a round $r_x < r$. 
\end{lemma}
\begin{proof}
The proof is by contradiction. Suppose the \Leader enters in state {\sf Detection} at round $r$, while \R is still alive. 
It follows that at round $r$ the predicate $FailedReport[\R]$ is verified by \Leader. 
Let $r_{last}$ be the last round, before $r$, in which \R and \Leader met; thus, $RLastMet[\R]=r-r_{last}>7((\#Meets[\R]+1)\sqrt{n}+\vTsteps)$. Where $\vTsteps$ is the value of the corresponding variable stored by \Leader at round $r$. 

Let $T$ be the interval between $r_{last}$ and $r$, that is $T=[r_{last},r-1]$. We now prove the fact (1): in interval $T$ the leader has done less than $\frac{|T|}{7}$ steps clockwise. Suppose the contrary, for each step clockwise of the leader the quantity 
$7((\#Meets[\R]+1)\sqrt{n}+\vTsteps)$ increases of $7$ units, this implies that after $\frac{|T|}{7}$ steps the quantity would be greater or equal than $|T|$, that is in contradiction with the triggering of $FailedReport[\R]=  |T| > 7((\#Meets[\R]+1)\sqrt{n}+\vTsteps)$.  

Now we will show that fact (1) is in contradiction with the hypothesis that \R is alive at round $r$.  By construction, at round $r_{last}$, \R starts moving counter-clockwise until it explores all nodes in sector $S_{\#Meets[\R]}$. If \R is not blocked by a missing edge, it will complete this exploration in 
at most $3(\#Meets[\R]+1)\sqrt{n}+3\vTsteps$ rounds;
then, \R will move clockwise toward \Leader, reaching it in at most $(\#Meets[\R]+1)\sqrt{n}+\vTsteps)$ rounds. 
The overall sum of the above rounds is $4((\#Meets[\R]+1)\sqrt{n}+\vTsteps)$, that is clearly less than  $|T|$.

Therefore, the only possible scenario left to analyze is when \R has been blocked during its movement (by hypothesis, \R cannot reach the black hole before round $r$). 
The number of rounds in which \R has been blocked during interval $T$ is at least  $3(\#Meets[\R]+1)\sqrt{n}+\vTsteps)$:  if \R moves for $4((\#Meets[\R]+1)\sqrt{n}+\vTsteps)$ rounds, it meets the leader in interval $T$ preventing $FailedReport[\R]$ to trigger. 

However, for every three rounds in which \R is blocked the \Leader moves one step. Therefore, in $T$ the leader did (at least) $(\#Meets[\R]+1)\sqrt{n}+\vTsteps)$ steps clockwise, this quantity is precisely $\frac{|T|}{7}$  and it is in contradiction with fact (1). 

From the above we have that \R cannot be alive at round $r$, and this prove our claim. 
\end{proof}

\begin{lemma}\label{leaderneverdies}
The \Leader does not enter the black hole. 
\end{lemma}
\begin{proof}
By Lemma~\ref{gooddec}, if \Leader goes into the {\sf Detection} state, then \R reached \BH; also, \R marked with a pebble the last safe visited node. Thus, \Leader will find the marked node and terminate
before entering \BH (predicate $marked$ in state {\sf Detection}). 
If \Leader does not go in  state {\sf Detection}, then it moves clockwise simulating a cautious walk with \A. Since \Leader waits for \A to return before visiting a new node, it follows that \Leader will never enter the black hole. 
\end{proof}

\begin{lemma}\label{firstenter}
Let $r$ be the first round in which one among \A or \R enters the black hole. We have $r={\cal O}(n^{1.5})$.
\end{lemma}
\begin{proof}
First note that, as long as \R does not reach the black hole, the set of nodes that are visited by \R and \A  are disjoint (not considering the black hole \BH): in fact, \A switches direction of movement only when \Leader goes in state {\sf Detection}; however, by Lemma~\ref{gooddec}, this transitions cannot occur as long as \R does not reach the black hole.  
Therefore, since at most one edge at the time might be missing, it follows that at most one among \R and \A may be blocked at each round. 

If \R is not blocked, it first explores sector $S_{0}$ of $\sqrt{n}$ nodes, then it goes back to \Leader (that is at most $n$ hops away), then it explores the new sector $S_1$ of $\sqrt{n}$, and so on. Therefore, its exploration costs a number of rounds that is upper bounded by:
$$ \sum^{\sqrt{n}}_{i=0} 3(i\sqrt{n} +n).$$

\noindent By immediate algebraic manipulation, we have that $3\sum^{\sqrt{n}}_{i=0} (i\sqrt{n} +n)\leq 6n\sqrt{n}$.

Now we observer that if  \A is free to move for $3n$ rounds, then it explores $n$ nodes. 
 Thus, if \A is not blocked for at least $3n$ rounds over an interval of $12n\sqrt{n}$ rounds, it will necessary reach the black hole \BH. However, if \A is blocked for $3n$ rounds, then \R is free to move for $12n\sqrt{n}-3n\geq 6n\sqrt{n}$ rounds, and it reaches \BH. From the above in the first  $12n\sqrt{n}$ \R or \A reaches the black hole, and thus the lemma follows. 
\end{proof}

\begin{lemma}\label{leader:terminate}
Let $r$ be the first round when an agent enters the black hole; then, \Leader terminates by round $r_t=r+{\cal O}(n^{1.5})$.  
\end{lemma}
\begin{proof}
Let $r$ be the first round when someone enters in the black hole, note that by Lemma~\ref{leaderneverdies} we can exclude that is the \Leader to enter in the black hole. 
Therefore, at round $r$ \A or \R reached the black hole. 

 We will now show that \Leader terminates by round $r'=r+{\cal O}(n^{1.5})$. We distinguish two cases:
\begin{itemize}
\item \A reaches the black hole \BH at round $r$. By construction, at round $r$ \Leader is at node $v$ neighbour  of \BH; let $e$ be the edge between $v$ and \BH. 
\Leader waits on $v$ until either: (1) the edge $e$ is not missing, or (2) \R fails to report. 
In case (1), predicate $FailedReport[\A] $ is verified  by \Leader, hence \Leader terminates. 

If (1) does not apply, we first show then there is a round $r_f > r$ in which \Leader detects that \R has failed to report, with $r_f \in {\cal O}(n \cdot \sqrt{n})$.  Since $e$ is missing, \R cannot be blocked; hence, it reaches \BH by round $r+6 \cdot n \cdot \sqrt{n}$ (refer to the argument for the bound on $r$ used in Lemma~\ref{firstenter}). Also, \R marks the node before the black hole with a pebble. 

By predicate $FailedReport[\R]$, \Leader detects this event (\R in \BH) by round $r_f=r+6 \cdot n \cdot \sqrt{n}+7((\#Meets[\R]+1)\sqrt{n}+\vTsteps)$.
We argue that $r_f=r+6 \cdot n \cdot \sqrt{n}+7((\#Meets[\R]+1)\sqrt{n}+\vTsteps)$ is ${\cal O}(n \cdot \sqrt{n})$: we have that $\vTsteps \leq n$, otherwise the \Leader would have entered the black hole that is impossible, see Lemma~\ref{leaderneverdies}); and  $\#Meets[\R] \leq n$, each times \R meets the leader it explores $\sqrt{n}$ nodes thus \R would enter in the black hole after at most $\sqrt{n}$ meetings.  

At round $r_f$, \Leader changes state, it reverts direction of movement, and starts moving towards the node that \R marked. If \Leader  does not reach the marked node within $3n$ rounds, then \A fails to report to \Leader (by predicate $FailedReportD$), and \Leader terminates. 
Otherwise, if \Leader reaches the marked node, it terminates as well. In both cases, the leader within ${\cal O}(n^{1.5})$ rounds, and the lemma follows.

\item  \R reaches the black hole at round $r$, while exploring sector $S_{J}$. By construction, \R marks the node $v$ before \BH with a pebble; also, $J=\#Meets[\R]$. 
Consider now round $r'=r+7((\#Meets[\R]+1)\sqrt{n}+\vTsteps$: if the \Leader does not terminate by round $r'$,  then it detects that \R failed to report (by predicate $FailedReport[\R]$). 

Since $\#Meets[\R] \leq n$ and $\vTsteps \leq n$, it follows that $r'= {\cal O}(n \cdot \sqrt{n})$. 
At round $r'$ \Leader  goes into {\sf Detection} state, and it switches direction of movement. Now, if \A reaches \BH by round $r'$, then the proof is similar to the previous case. Otherwise, 
if \Leader is never blocked by a missing edge, it  reaches $v$ at most $n$ rounds after switching direction, it hence terminates, and the lemma follows again. 

Thus, let us  consider the case when \Leader is blocked by a missing edge before $n$ rounds after $r'$.  We assume that \A does not enter in the black hole by round $r'$, otherwise the proof is equal to previous case. 

 While the \Leader is blocked, \A has enough time to trigger the predicate  $\neg\pSees[\Leader]$, and to go {\sf SearchLeader} state;
in this state, \A moves toward the \Leader. When \A and \Leader meet, \A  uses the value of  $\#Meets[\R]$ to identity the node of sector $S_{\#Meets[\R]}$ that has to be explored (this is done by updating variable $nextTarget$). At this point \A, in state  {\sf Detection1}, moves until it explores, from clockwise direction, the first node in the dangerous sector. If such node is not the black hole \A switches to state {\sf Return1}, it goes back to \Leader, and upon meeting it goest to {\sf Detection2} updating its target node in the dangerous sector. This oscillation between states {\sf Return1} and {\sf Detection2} is iterated until \A enters in the black hole. 

We claim that, if \A is not blocked, then \A enters in the black hole in at most $2n \sqrt{n}$ rounds:  in state {\sf Detection2} \A explores a new node in $S_{\#Meets[\R]}$, then it switches to state {\sf Return1} and meets the \Leader; this process lasts at most
$2n$ rounds. Since $S_{\#Meets[\R]}$ contains $\sqrt{n}$ rounds the claim follows. 
Once \A entered in the black hole, the \Leader will terminate after additional $3n$ rounds (see predicate $FailedReportD$). \end{itemize}
\end{proof}

\begin{lemma}\label{leaderisalwayscorrect}
If the \Leader terminates, it can correctly locate the position of the black hole.  
\end{lemma}
\begin{proof}
The proof proceeds by case analysis on the terminating conditions of Algorithm~\ref{alg-dpl}. 
\begin{itemize}
\item \Leader terminates in state {\sf TerminateA}. This case occurs when $FailedReport[\A]$ is triggered and \Leader is in state either  {\sf Init} or {\sf Cautious}. 
In both states, \Leader is waiting for \A to return; also, $FailedReport[\A]$  triggers if the clockwise edge $e$ is present for two rounds (not necessarily consecutive), and $meeting[\A]$ has not been verified. 
The first round \A moves to a neighbor node using $e$, and it goes to {\sf Return} state. Thus, as soon as edge $e$ appears again the second time, an alive \A would get back to the node where the \Leader is, preventing the triggering
of $FailedReport[\A]$. Therefore, the only possibility left, is for \A to be lost in the black hole. Hence, the \Leader correctly terminates and correctly detects the position of \BH.

\item \Leader terminates in state {\sf TerminateR}.  First, \Leader can go in the {\sf TerminateR} state only from state {\sf Detection}. By Lemma~\ref{gooddec}, \Leader goes in the {\sf Detection} state only after
\R reaches the black hole. By construction, when \R reaches the black hole, it leaves a pebble in the clockwise neighbour of \BH (\R performs a cautious exploration in both states {\sf Init} and {\sf Bounce}); also, \R is the only agent that leaves a pebble. 
Therefore, 
when \Leader finds a marked node, it necessarily is the node marked by \R.
Hence, \BH is the counter-clockwise neighbour of the marked node, and thus \Leader terminates correctly. 

\item \Leader terminates in state {\sf TerminateAD}. First, this state is reachable only from the {\sf Detection} state; also, by Lemma~\ref{gooddec},  \Leader goes in the {\sf Detection} state only after \R reaches the black hole. Let $r$ be the round when the \Leader goes to the {\sf Detection} state. Note that the \Leader resets variable $\#Meets[\A]$ to $0$. 

If \Leader reaches the node marked with a pebble by \R, then it would terminate in state {\sf TerminateR}. (see previous case). Therefore, let us assume that 
\Leader never reaches the node marked by \R; that is, \Leader is blocked by a missing edge starting from round $r_b=r+\delta$ on, with $\delta < n$. 

First, we will show that $FailedReportD$ cannot trigger as long as \A is alive. If \A is alive at round $r_b$  (i.e., it did not reach \BH), it follows that \A cannot be blocked from round $r_b$ on; hence, at most by round $r_b+2$, it goes to the {\sf searchLeader} state. Now, in at most $n$ rounds, \A meets \Leader, it computes the position of the 
dangerous sector, and it starts oscillating with a period of $2n$ rounds. Therefore, predicate $FailedReportD$ cannot trigger as long as \A is alive, and when \A enters in the black hole after round $r_b$ the \Leader knows the position of \BH (Recall that, \Leader is always aware of the node that \A is exploring in the dangerous sector, by accessing the value of $\#Meets[\A]$).

Let us now consider the case when \A is not alive at round $r_b$: this case occurs only if \A reaches \BH, and \BH is the clockwise neighbour of the node where \Leader is when it goes to the {\sf Detection} state. By construction of the algorithm before round $r+1$, \A and \Leader do not occupy the same node only when \A explores a new node (state {\sf Return}).  Hence, since \A is not alive, predicate $FailedReportD$   triggers at round $r+3n$, and $\#Meets[\A]=0$. Thus, \Leader correctly locate the position of  the black hole: it is in the clockwise neighbour of the node visited by the \Leader at round $r$ (that is the node where the \Leader switches to {\sf Detection} state). 
\end{itemize}

Therefore, in all cases, the lemma follows.
\end{proof}

\begin{theorem}\label{th:dp}
Consider  a dynamic ring  ${\cal R}$ with three colocated agents in the Pebble model. 
Algorithm \DP solves the \BHS\   with ${\cal O}(n^{1.5})$ moves and ${\cal O}(n^{1.5})$ rounds. 
\end{theorem}
\begin{proof}
By Lemmas~\ref{firstenter} and~\ref{leaderneverdies}, it follows that in at most ${\cal O}(n^{1.5})$ rounds, either \R, or \A, or both, reach the black hole. 
At this time, by Lemma~\ref{leader:terminate}, \Leader terminates in  ${\cal O}(n^{1.5})$ rounds; also, by Lemma~\ref{leaderisalwayscorrect}, \Leader correctly identifies the position of the black hole, thus solving \BHS.
It is also clear that agents in state \A and \R cannot terminate by algorithm design, they have no terminating state in their algorithms. Thus, they cannot terminate incorrectly. 

Finally the bound on the number of moves derives directly from the bound on the number of rounds needed to terminate, and from the fact that the number of agents is constant.   
\end{proof}

We can conclude that:
\begin{theorem}\label{ic3optimal}
Algorithm \DP is    size-optimal with optimal cost and time.\end{theorem}

\section{Scattered Agents}\label{sec:scattered}

In this section we study the \BHS\ in rings when agents are {\em scattered}, i.e., they start from different home-bases. As we did for the colocated case we distinguish between Endogenous and Exogenous communication mechanisms.

\subsection{Impossibility with Endogenous Communication}
When the agents are scattered,  any number of them, even equipped with the stronger endogenous mechanism (i.e., F2F model), cannot solve \BHS on rings of arbitrary size. 
Interestingly,  the following theorem holds also for static rings. 

\begin{theorem}\label{scattered:idsimpossible}
A constant number of scattered agents in the F2F model cannot solve \BHS on a static ring 
of arbitrary size $n$, even if they
have distinct IDs.  
\end{theorem}
\begin{proof}
Let us consider $k={\cal O}(1)$, and ${\cal A}$ be an algorithm that correctly solves the problem. The proof is by contradiction; we will show that there exists an initial configuration $C$ of the $k$ agents on a ring of a proper size $n$ that makes ${\cal A}$ fail.
We will construct the configuration $C$ in such a way that no two agents meet.

 Let $id_1,id_2,\ldots,id_k$ be the IDs of agents, with a small abuse of notation we will use $id_i$ to denote also the agent with $id_i$.

We consider the behaviour of agent $id_i$ until round $r$ in a run where it executes the algorithm ${\cal A}$, and it does not meet any agent. 
 
 Let $DLeft(id_i,r)$ be the furthest distance travelled in the left direction from its initial homebase by agent $id_i$ until round $r$, and let $DRight(id_i,r)$ be the analogous for the right direction. Moreover, we define as $D(r)$ the maximum distance in any direction travelled by an agent until round $r$. 
 
 We say that agent $id_i$ is $L$-bounded if $\exists L \in \mathbb{N}$ such that $\forall r > 0$ we have $DLeft(id_i,r) < L$ and $DRight(id_i,r) < L$. Note that, an $L$-bounded agent, when alone on a ring of a sufficiently large size, will be perpetually confined on a constant set of nodes, never visiting the entire ring. Let $B$ be the set of agents that are $L$-bounded for some $L$. 
 Now let us focus on the set $UB$ of agents in $A$ that are not bounded, let $t=|UB|$. We will show that we can place agents in $UB$ on a ring $R$ of a size $n$ (value that we will specify later) such that: (1) all agents placed will enter in the black hole, and (2) no two agents meet. The placement is constructed inductively. At each step we place on more agent from $UB$ on $R$, in such a way: the newly placed agent $id_i$ will enter in the black hole at a round $r_i$, by round $r_i$ all the agents previously placed also entered the black hole, and no meeting between any pair of agents happened. 
 \begin{itemize}
 \item First agent: Let $r_1$ be the first round at which some agent, w.l.o.g $id_1 \in UB$ has $DLeft(id_1,r_1) =1$ or $DRight(id_1,r_1) =1$, note that it also holds $D(r_1)=1$. It is clear that if we position agent $id_1$ at distance $d_1=D(r_1)$ from the black hole \BH in the correct direction (either left or right according to whether $DLeft(id_1,r_1) =1$ or $DRight(id_1,r_1) =1$) it will enter in the black hole at round $1$. 
 
 \item Second agent: The second agent will be placed at distance $d_2=d_1+2(D(r_1)-1)+1$ from the \BH. We argue that if we place any agent at this distance either in the left or right direction, then it will not meet any of the previous agents by round $r_1$: no agent does more than $D(r_1)$ step before round $r_1$, and at round $r_1$ all agents placed before will enter in the \BH. 
 
 Let $r_2 >r_1$ be the first round at which some agent, w.l.o.g. $id_2 \in UB \setminus \{id_1\}$ has $DLeft(id_2,r_2) =d_2$ or $DRight(id_2,r_2) =d_2$, note that it also holds $D(r_2)=d_2$. We stress that such a round exists, since no agent in $UB$ is $L$-Bounded.
 
Now, if we place agent $id_2$ at distance $d_2$ from the black hole \BH in the correct direction it enters in the black hole by round $r_2$.

\item Third agent: The third agent will be placed at a distance $d_3=d_2+2(D(r_2)-1)+1$ from \BH. We argue that if we place any agent at this distance from \BH either in the left or right direction, no agent will meet by round $r_2$: no agent does more than $D(r_2)-1$ steps before round $r_2$, the other agents are placed at a distance that is larger than $2D(r_2)-1$ and by round $r_2$ any previously placed agent entered in \BH. 

Let $r_3>r_2$ be the first round at which some agent, w.l.o.g. $id_3 \in UB \setminus \{id_1,id_2\}$ has $DLeft(id_3,r_3) =d_3$ or $DRight(id_3,r_3) =d_3$, note that it also holds $D(r_3)=d_3$. We stress that such a round exists, since no agent in $UB$ is $L$-Bounded.

Now, if we place agent $id_3$ at distance $d_3$ from the black hole \BH in the correct direction it enters in the black hole by round $r_3$.  

\item $t$-th agent: The $t$-th agent will be placed at a distance $d_t=d_{t-1} +2(D(r_{t-1}-1))+1$ from the \BH. Let $r_t > r_{t-1}$ be the first round in which agent $id_t$ has $DLeft(id_t,r_t) =d_t$ or $DRight(id_t,r_t) =d_t$. We stress that such a round exists, since $id_t$ is not $L$-Bounded.

 By construction by round $r_{t-1}$ there can be no meeting between agent $id_t$ and agents $id_1, \ldots, id_{t-1}$: no agent travels a distance $D(r_{t-1})$. Since until round $r_{t-1}$ all agents $id_1, \ldots, id_{t-1}$ act has they are alone, then
by inductive hypothesis by round $r_{t-1}$ all agents $id_1, \ldots, id_{t-1}$ entered in the \BH without any meeting between them. It is now immediate to see that at round $r_t$ also $id_t$ enters in the \BH. 

 \end{itemize} 

Now we have to place the agents that are in $B$. Let $L_{M}$ be an upper bound on the nodes explored by any agent in $B$. It is clear that if we place agents in $B$ at a distance $2L_{M}+1$ from each other, they will never meet. 
Therefore, if we place all agents in $B$ on ring $R$ starting at a distance $3d_t+L_{M}+1$ from \BH at regular intervals of size $2L_{M}+1$, then they will never meet with any of the agents in $UB$. 
In order for the above placement to be possible we take a ring $R$ of size $n \geq 2(3d_t + |B|2(L_{M}+1))+1$. 
It is clear that ${\cal A}$ cannot be a correct algorithm on $R$ since an agent enters the \BH without meeting anyone else, or the agent is perpetually confined in a portion of the ring, and it does not meet anyone else.
In both cases, no agent knows the position of \BH. 
\end{proof}

\subsection{Exogenous Communication Mechanism} 
Fortunately, any exogenous mechanism circumvents the impossibility of Th.~\ref{scattered:idsimpossible}. In this section we focus on such mechanism, as we remarked in Section~\ref{whiteworld}, even the weaker exogenous mechanism (i.e., pebbles) allow for symmetry breaking and communication of messages.

\subsubsection{A Quadratic Lower Bound on the Cost and Time of Size-optimal Algorithms.}
Interestingly, we can show a quadratic lower bound on the number of moves and rounds of any size-optimal algorithm that solve that \BHP with scattered agents; the bound holds even if agents have IDs and whiteboards are present.

\begin{theorem}\label{scat:lb}
In  a dynamic ring ${\cal R}$   with whiteboards, any algorithm ${\cal A}$ for \BHS with three scattered agents with unique IDs requires $\Omega(n^{2})$ moves and $\Omega(n^{2})$ rounds.
\end{theorem}
\begin{proof}
The proof is by contradiction. Let ${\cal A}$ be a sub-quadratic algorithm that solves the \BHP; and let $a$, $b$, and $c$ be the tree agents. Suppose the agents have unique IDs, and, without loss of generality, let $c$ be the first agent that moves. 

Let us assume an initial configuration where $c$ is initially placed on node $v_{c}$, neighbour of \BH, and where $a$ and $b$ are on two neighbours nodes, $v_a$ and $v_b$, at distance $n/2$ from $v_c$. 
 Furthermore, let us assume that $c$ is placed in such a way that when it moves, it immediately enters \BH. Without loss of generality, we may assume that $r=0$ is the round and that $v_c$ is the counter-clockwise neighbor of \BH. 

Let $U_r$ be the partition of nodes explored by agents $a$ and $b$ at the end of round $r$. By Observation~\ref{madebyus}, agents $a$ and $b$ may explore a node outside $U_r$ only if they try to traverse at the same round both the edges crossing the cut $U_r$ and $V \setminus U_r$. Let $e_c$ be the clockwise edge incident in $U_r$, and $e_{cc}$ be the counter-clockwise one. Edge $e_c$ might always be missing, thus preventing the agents from crossing it. Therefore, \BH might only be reached by its clockwise neighbor, and node $v_c$ will never be explored. 

Let $r$ be a round such that $U_{r-1} \subset U_r$. If $U_{r-1}$ is ${\Theta}(n)$, in order for both agents to explore $U_r$, they have to perform ${\Theta}(n)$ moves.
Therefore, since by hypothesis ${\cal A}$ is sub-quadratic, there must exists an agent, say $b$, that explores at least two nodes, say $v_1$ at round $r_1$ and $v_2$ at round $r_2$, such that it (1) does not communicate with $a$ between the two explorations and (2) both $a$ and $b$ visits $o(n)$ disjoint nodes between $r_1$ and $r_2$. 

Note that, since the initial configuration is arbitrary and $a$ and $b$ never received any information communicated from $c$, the position of $v_1$ and $v_2$ does not depend on the position of \BH and $v_c$. Therefore, there can exist two initial configurations $C_1$ and $C_2$, such that $v_1=$\BH in $C_1$ and $v_2=$\BH in $C_2$. 
Since $b$ reaches \BH by round $r_2$, $a$ is the only agent that can disambiguate between the two configurations. However, $a$ might be blocked indefinitely on a set of nodes that
was never visited by $b$ after round $r_1$ (see Observation~\ref{madebyus}, and recall that at round $r_2$ agent $a$ was trying to traverse edge $e_c$). 
Consequently, $a$ is not able to distinguish $C_1$ and $C_2$; thus, ${\cal A}$ cannot be correct, having a contradiction.
Finally, the bound on the number of rounds derives immediately from the one on moves and the fact that there is a constant number of agents.. 
\end{proof}

The above theorem shows the cost optimality of the size-optimal algorithm \GL described in Section~\ref{scatter:threeagents}. 

\subsubsection{An Optimal Exogenous Algorithm: \GL} \label{scatter:threeagents}
In this section, we describe an algorithm to solve the problem with $k=3$ anonymous agents using pebbles.
We name the algorithm \GL. \GL works in two phases:

\begin{itemize}
\item {\bf Phase 1:} In the first phase agents move clockwise using pebbles to implement a cautious walk. If they meet they synchronise their movements such that at most one of them enters in the black hole. The Phase 1 lasts until all three agents meet or $9n$ rounds have passed. 
We will show that at the end of this phase we have either:
\begin{itemize}
\item  (1) three agents are on the same node or on the two endpoint nodes of the same  edge. In this case we say that agents gathered; 
\item (2) the counter-clockwise neighbour of the black hole has been marked, at most one agent is lost, and the two remaining agents are gathered (that is they are on the same node, or on two endpoints of an edge);
\item  (3)one agent is lost in the black hole, the counter-clockwise neighbour of the black hole has been marked. For the remaining two agents we have that either both terminated locating \BH, or at least one terminated, and the other is still looking for the \BH. In Phase 2 this last agent will either terminate or it will be blocked forever (Note that in either case the problem has been solved, since at least one agent located the \BH). 
\end{itemize}

\item {\bf Phase 2:} The second phase starts after the previous one, and relies on the properties enforced by the first phase. In particular, if at the beginning of this phase three agents are on the same node, they start algorithm \RT. Otherwise, if there are at least two agents on the same node, they act similarly to \R and \Leader in \RT.  If none of the above applies, then two agents are on the two endpoints of a missing edge, or a single agent remained and still has to discover the \BH. This scenario is detected by a timeout strategy: upon a timeout, an agent starts moving clockwise looking for the node marked during Phase 1. If two agents meet during this process, they act similarly to \R and \Leader in \RT. Otherwise, in case a single agent is still active, the agent will either reach the marked node (and terminate correctly) or it will be blocked forever on a missing edge. We remark that in this last case there has been an agent correctly terminating in Phase 1, and thus \BHP is still correctly solved.  
\end{itemize}

\paragraph{Detailed Description.}
The pseudocode of {Phase 1} is reported in Algorithms~\ref{alg:scatteredchirality},~\ref{alg:scatteredchiralityfollower}, and~\ref{alg:scatteredchiralityexplorer}; and Phase 2 in Algorithms~\ref{alg:p2chirscat} and~\ref{alg:p2t}. 
In the pseudocode, we use predicate $\#A=x$ that is verified when on the current node there are $x$ agents. Initially, all agents have role {\sc Anon} (that is, agents are anonymous). As discussed in Section~\ref{par:pebble}, if one agent is trying to remove the mark from a node, during a cautious walk, it will not trigger the meeting predicate with any of the other agents. 
\medskip

 \begin{figure}
\center
  \begin{subfigure}[b]{0.6\textwidth}
  \center
    \includegraphics[width=\textwidth]{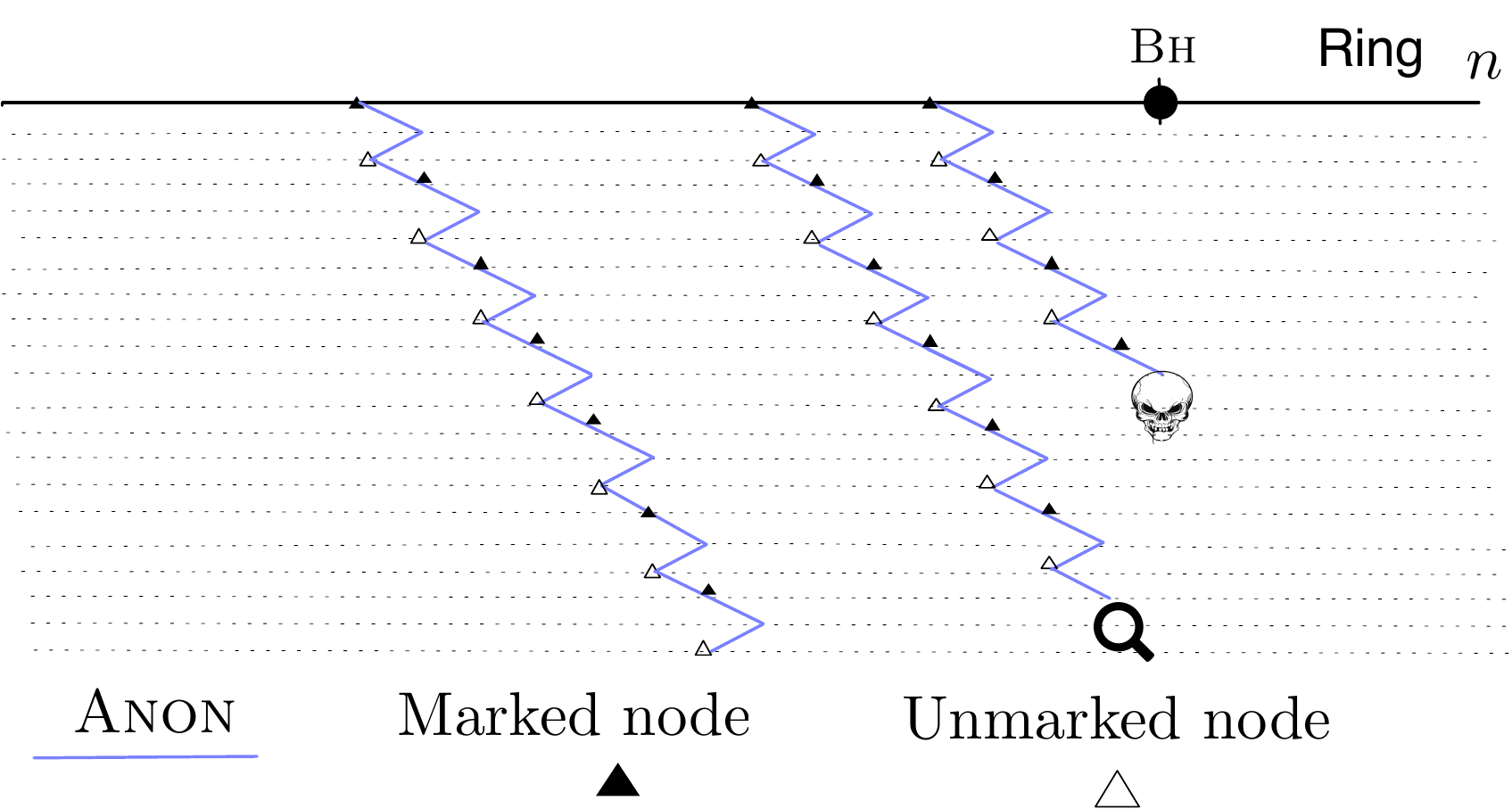}
    \caption{Phase 1: run where \BHP is solved.}
    \label{scat:phase1c1}
  \end{subfigure}
\\
  \begin{subfigure}[b]{0.6\textwidth}
  \center
    \includegraphics[width=\textwidth]{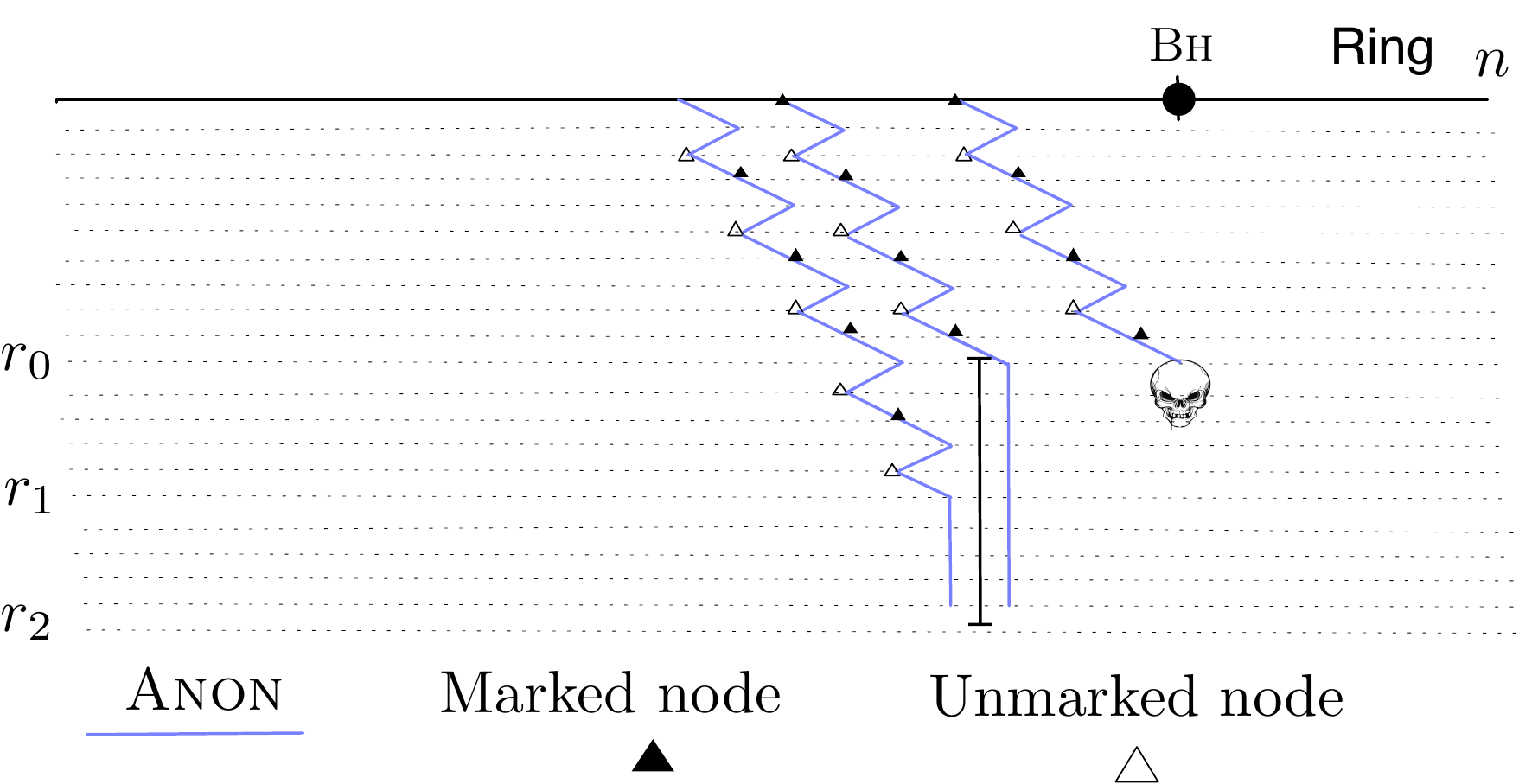}
    \caption{Phase 1: run where two agents gather. At round $r_0$ the rightmost agent enters in the black hole, while
    the middle agent is blocked. At round $r_1$ the two remaining agents gathers.}
    \label{scat:phase1c3}
  \end{subfigure}
\\
    \center
  \begin{subfigure}[b]{0.6\textwidth}
    \includegraphics[width=\textwidth]{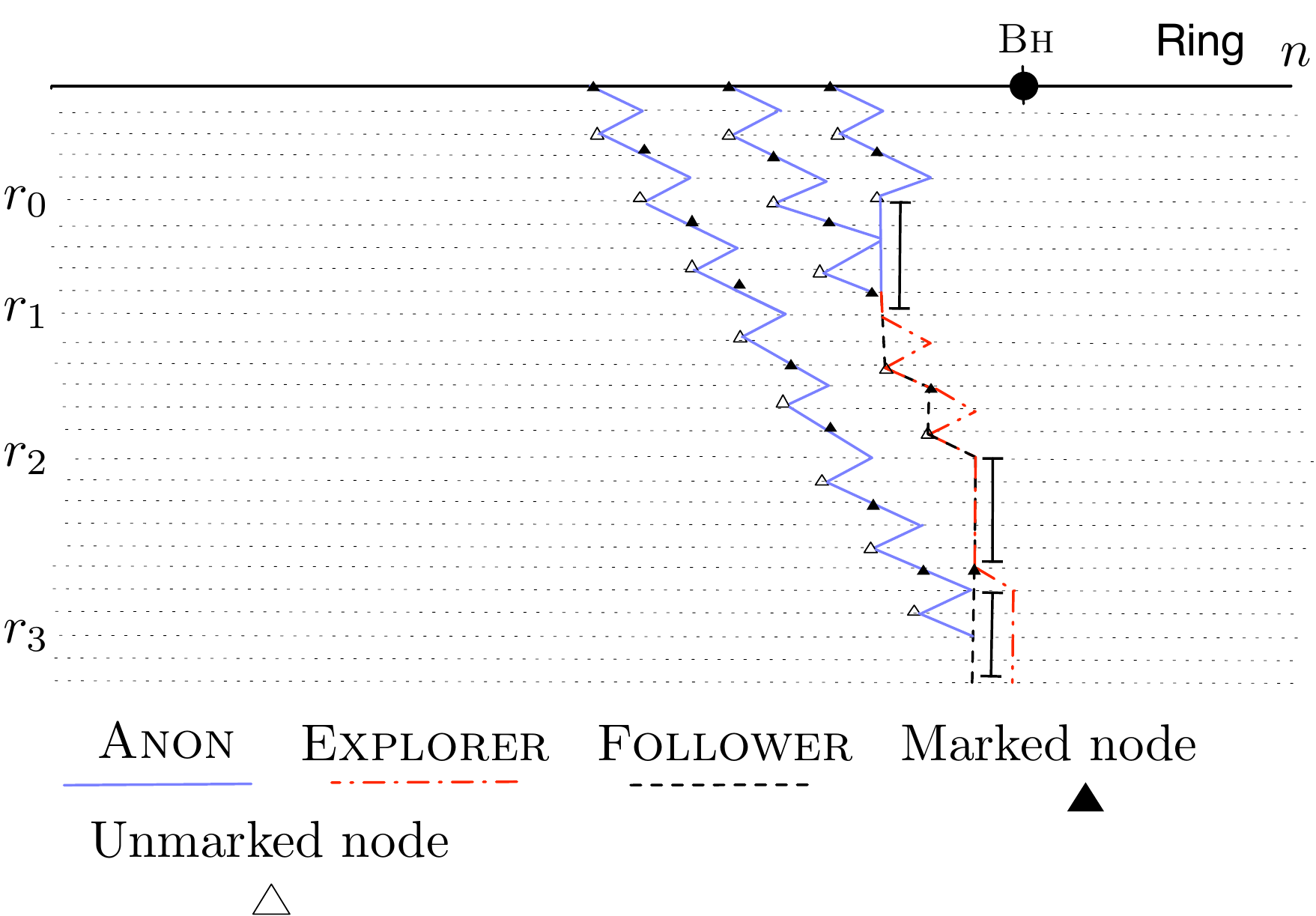}
    \caption{ Phase 1: run where three agents gather. At round $r_0$ the rightmost agent is blocked. At round $r_1$ two agents
    meet creating a pair {\sc Explorer}-{\sc Follower}. At round $r_2$ the pair is blocked and the leftmost agent is able to catch up.
    At round $r_3$ the tree agents gathered. Note that the meeting predicate of the {\sc Anon} with the {\sc Follower} triggers at round $r_3$
    and $r_3-2$: when an agent is cautious exploring it cannot meet other agents if it has to unmark a node.}
    \label{scat:phase1c2}
  \end{subfigure}

  \caption{Example of runs for Phase 1 of \GL.}\label{fig:scat}
\end{figure}

{\bf Phase 1:} 
The first phase lasts for at most $9n$ rounds (refer also to the examples in Figure~\ref{fig:scat}). The agents start in state {\sf Init} of Algorithm~\ref{alg:scatteredchirality}:  
each agent walks cautiously clockwise for $9n$ rounds. 
If an agent reaches a marked node, then it waits until the next node can be deemed as safe or unsafe (see state {\sf Wait}). 

If in the marked node the incident clockwise edge is present and the agent that marked the node does not return, then the next node is the black hole (the agent terminates by triggering predicate $NextUnsafe$). 

If two {\sc Anon} agents meet on the same node (predicate $meeting[${\sc Anon}$]$), they synchronise their movements such that they will never cross an edge leading to a possibly unsafe node in the same round. Specifically, the agents enter in the synchronisation state {\sf Two}, where one agent becomes {\sc Follower} (Algorithm~\ref{alg:scatteredchiralityfollower}) and the other becomes {\sc Explorer} (Algorithm~\ref{alg:scatteredchiralityexplorer}). The role of {\sc Explorer} is to visit new nodes, while {\sc Follower} just follows {\sc Explorer} when a node is safe, (this is similar to \Leader and \A in \RT). 
If the remaining {\sc  Anon} agent meets with the {\sc Follower}, it will mimic the behaviour of the {\sc Follower} (predicate $meeting$[Follower] in state {\sf Init} and state {\sf Copy}). 
Finally, if the three agents meet on the same node, Phase 1 terminates (see predicate $\#A=3$ in all states). In any case, at the end of round $9n$, this phase ends.

\noindent In Section~\ref{sec:correctness}, we will show that, if in Phase 1 all the alive agents have not localised the \BH, then either:
\begin{itemize}
\item three agents gathered: either three agents are on the same node, or two agents are on a node $v$ and the third agent is blocked on the clockwise neighbour $v'$ of  (the marked) node $v$; or
\item the counter-clockwise neighbour of the black hole has been marked, at most one agent is lost, and the two remaining agents are gathered. The two agents are either on the same node, or on two different neighbours node and one of them has marked the node where the other resides. 
\item the counter-clockwise neighbour of the black hole has been marked, at most one agent is lost. One agent correctly terminated, while the other is still looking for the \BH. 
\end{itemize}


\vspace{0.3cm}
   \begin{algorithm*}
\footnotesize

\begin{algorithmic}[1]
\State  Predicates Shorthands: $NextUnsafe =\vEtime > \vMtime[C]$
\State $NextSafe=$ the agent that marked the node returned.
\State States: \{{\sf Init}, {\sf Wait}, {\sf EndPhase1}, {\sf Terminate}, {\sf Copy}\}.
\AtState{Init}
    \State \Call{CautiousExplore}{\dRight$|$ $\vTtime = 9n \lor \#A=3$: {\sf EndPhase1}; $marked$: {\sf Wait}; $meeting$[{\sc Anon}]: {\sf Two}; $meeting$[{\sc Follower}]: {\sf Copy}}
    \AtState{Wait}
    \State \Call{Explore}{\dNil$|$ $\vTtime = 9n \lor \#A=3$: {\sf EndPhase1}; $NextUnsafe$: {\sf Terminate}; $NextSafe:$ {\sf Init}}

     \AtState{Two}
      \State assign to yourself a role in $\{${\sc Follower}, {\sc Explorer}$\}$
	\State Execute the corresponding Algorithm, that is Alg.~\ref{alg:scatteredchiralityfollower} in state {\sf WaitFollower}, or Alg.~\ref{alg:scatteredchiralityexplorer} in state {\sf Explore}.
	     \AtState{Copy}
	\State set your role to {\sc Follower} with the same state of the other agent.
	
	\AtState{EndPhase1}
	 \State take the role of {\sc Anon}
          \State starts Phase 2 by entering state {\sf InitP2} of Alg.~\ref{alg:p2chirscat}. 
     \AtState{Terminate}

     \State terminate, \BH is the next node in clockwise direction.
\END
\end{algorithmic}
\caption{\GL; Phase 1 - Algorithm for scattered agents - {\sc Anon} \label{alg:scatteredchirality}}

\end{algorithm*}

   \begin{algorithm*}[h]
\footnotesize
\begin{algorithmic}[1]
\State States: \{{\sf Explore}, {\sf Back}, {\sf MoveForward}, {\sf EndPhase1}, {\sf Terminate}\}. \Comment{  {\sf Terminate} and {\sf EndPhase1} as in Algorithm~\ref{alg:scatteredchirality}}
\AtState{Explore}
\If{current node is not marked}
	\State mark current node \label{explorer:mark}
    \State \Call{Explore}{\dRight$|$ $\vTtime = 9n \lor \#A=3$: {\sf EndPhase1}; $\vEsteps >0:$  {\sf Back}}
    \Else
        \State \Call{Explore}{\dNil$|$ $\vTtime = 9n \lor \#A=3$: {\sf EndPhase1}; $NextUnsafe$: {\sf Terminate}} \Comment{If the node is marked we have to wait to see if it is safe to move}
    
    \EndIf
    \AtState{Back}
    \State \Call{Explore}{\dLeft$|$ $\vTtime = 9n \lor \#A=3$: {\sf EndPhase1}; $\vEsteps > 0:$ {\sf MoveForward}}
        \AtState{MoveForward}
        	\State unmark current node
    \State \Call{Explore}{\dRight$|$ $\vTtime = 9n \lor \#A=3 $: {\sf EndPhase1}; $\vEsteps > 0:$ {\sf Explore}}
   
\END
\end{algorithmic}
\caption{\GL; Phase 1 - Algorithm for {\sc Explorer} \label{alg:scatteredchiralityexplorer}}
\end{algorithm*}

   \begin{algorithm*}
\footnotesize
\begin{algorithmic}[1]
\State States: \{{\sf WaitFollower}, {\sf Follow}, {\sf EndPhase1}\}. \Comment{  {\sf EndPhase1} as in Algorithm~\ref{alg:scatteredchirality}}
\AtState{WaitFollower}
    \State \Call{Explore}{\dNil$|$ $\vTtime = 9n \lor \#A=3$: {\sf EndPhase1},  $Meeting[Back]$: {\sf Follow} }
    \AtState{Follow}
    \State \Call{Explore}{\dRight$|$ $\vTtime = 9n \lor \#A=3$: {\sf EndPhase1}, $\vEsteps > 0$: {\sf WaitFollower} }
   
\END
\end{algorithmic}
\caption{\GL; Phase 1 - Algorithm for {\sc Follower} \label{alg:scatteredchiralityfollower}}
\end{algorithm*}

\bigskip

 \begin{figure}
\center
  \begin{subfigure}[b]{0.45\textwidth}
  \center
    \includegraphics[width=\textwidth]{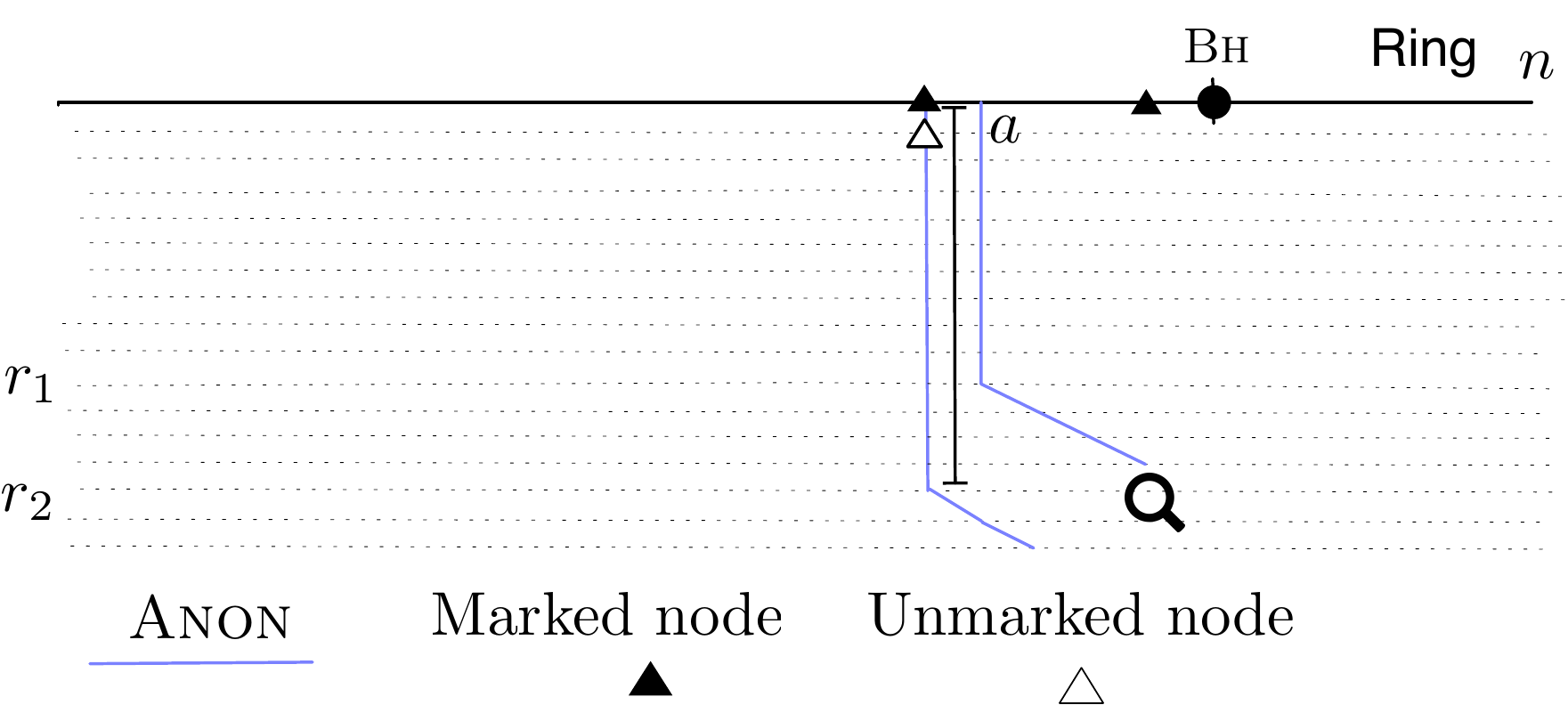}
    \caption{Phase 2: case of two agents gathered on the same edge. At round $r_1$ agent $a$ triggers the timeout and goes in state {\sf Forward} starting moving clockwise.
    At round $r_2$ it finds the marked node and terminates. Note that the other agent is either unblocked and thus reach the marked node and correctly terminate or it 
    waits forever on a missing edge.}
    \label{scat:phase2c1}
  \end{subfigure}
  \,\,
  \begin{subfigure}[b]{0.45\textwidth}
  \center
    \includegraphics[width=\textwidth]{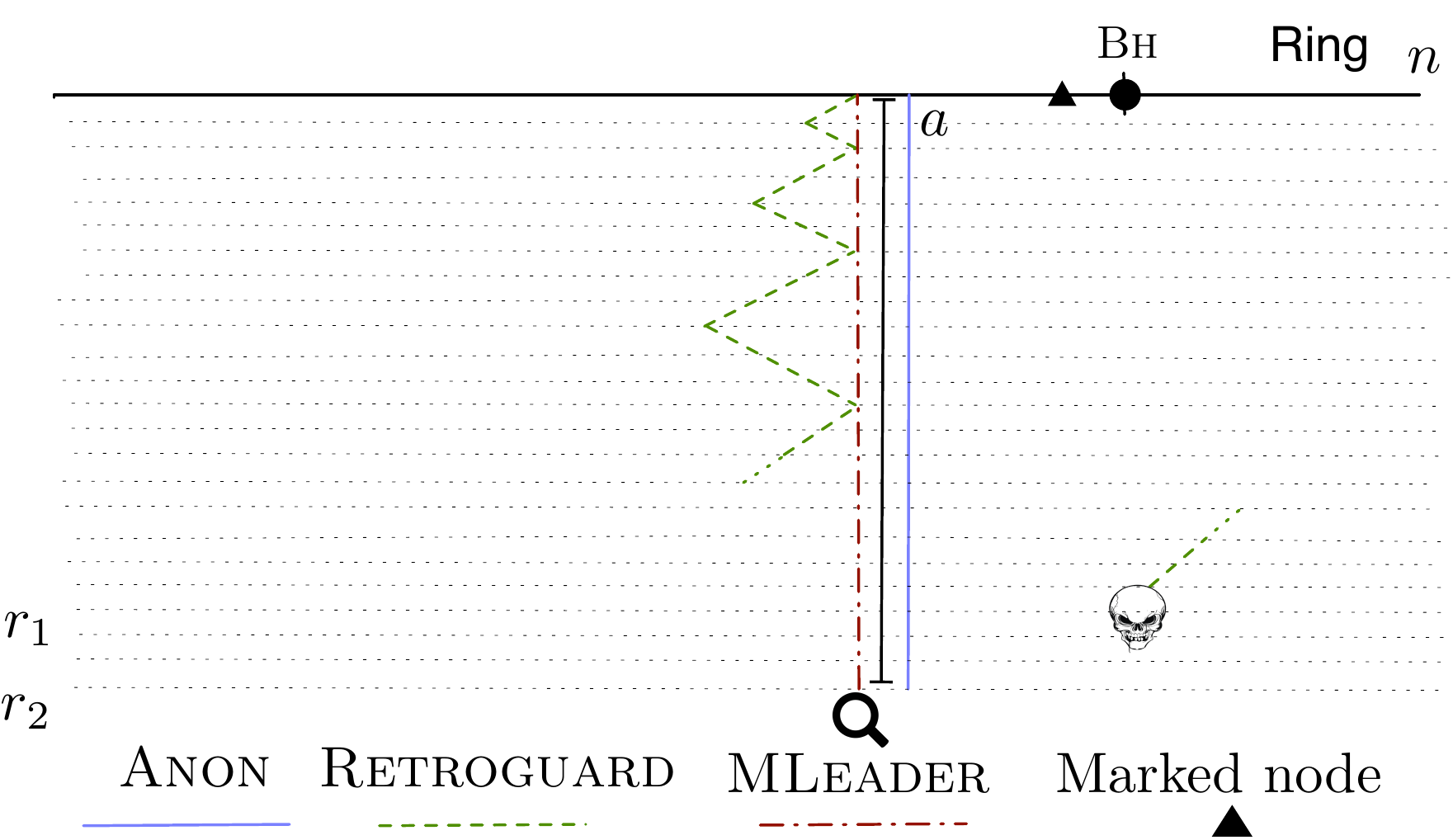}
    \caption{Phase 2: case of three agents gathered, two on the same node the other on the neighbour. At round $r_1$ the \R enters in the black hole. 
At round $r_2$ the {\sc MLeader} detects that \R failed to report an terminates. Note that this happens before $a$ timeouts.}
    \label{scat:phase2c2}
  \end{subfigure}
  \caption{Example of runs for Phase 2 of \GL.}\label{fig:p2}
\end{figure}

{\bf Phase 2:} 
The agents start in {\sf InitP2} state of Algorithm~\ref{alg:p2chirscat}: here, several checks are executed to understand how Phase 1 ended and to orchestrate the behaviour of the agents. In more details:
\begin{itemize}

\item Each agent checks if there are other agents on the same node:
in case there are two agents, they get the roles of \R and {\sc MLeader} (their behaviour is similar to \R and \Leader in \RT).
If there are three agents, they start algorithm \RT. 
\item Otherwise, the agent checks if it is missing the pebble. If this is the case, then the agent was blocked while trying to recover its pebble at the end of Phase 1: it starts moving counter-clockwise for $4 \cdot n^2$ rounds, and then it goes in the {\sf Foward} state. Intuitively, this move has the following goal: 
if there are two agents on the counter-clockwise node, then they have role \R and {\sc MLeader} and in $4 \cdot n^2$  rounds they have enough time
to find the black hole. Otherwise, if there is just one agent on the clockwise node or if there is no one, this timeout avoid that the agent is blocked forever on a missing edge. 

\item If the agent is on a marked node, then it waits there until either it meets the agent that marked the node, or $4 \cdot n^2$ have passed.
If they meet, they can break the symmetry by getting the roles of \R and {\sc MLeader}; otherwise, if the timeout triggers, the agent goes in state {\sf Forward}. 

\item If none of the above applies, the agent goes in state {\sf Forward}.
\end{itemize}

We now specify the behaviour of the agents: 
\begin{itemize}
\item Agent {\sc MLeader} and agent \R. The  {\sc MLeader} moves clockwise, while \R acts as in Algorithm \RT. 
If \R fails to report, {\sc MLeader} identifies the black hole and terminates. Finally, if {\sc MLeader} and an agent that is not \R meet, then this new agent takes the role of \A and {\sc MLeader} the role of \Leader, and they behave exactly as in Algorithm \RT (predicate $meeting[Leader] $ and state {\sf BeAvanguard} for the anonymous agent; and predicate \pMeeting[Anon] and state {\sf GoToCP} for the \Leader). The only caveat in this case, is that {\sc MLeader} keeps the value of variable $\#Meets[\R]$ when switching to \Leader.
\item Agent in state {\sf Foward}. In state {\sf Foward} an agent moves in the clockwise direction. If it reaches a marked node, then it discovered the black hole and the agent terminates.  If two agents in state {\sf Forward} meet, they break the symmetry, by getting the roles of {\sc MLeader} and \R. 

\end{itemize}

 \begin{algorithm*}
\footnotesize
\begin{algorithmic}[1]
\State  Predicates Shorthands: $NextUnsafe =\vEtime > \vMtime[C]$
\State States: \{{\sf InitP2}, {\sf BeAvanguard}, {\sf Terminate}, {\sf BreakSimmetry}\}.
\AtState{InitP2}
\If{$ \#A>1$}
\State go to state {\sf BreakSimmetry}
\ElsIf{my pebble is missing}
\State \Call{Explore}{\dLeft$|$$meeting$[{\sc Anon}]: {\sf BreakSimmetry}; $meeting$[{\sc MLeader}]: {\sf BeAvanguard};  
$\vEsteps >0$: {\sf Forward} ; \vTtime $> 4n^2$: {\sf Forward}} \label{p2:b1}
\ElsIf{the current node is marked}
\State take the pebble \label{p2:a1}
\State \Call{Explore}{\dNil$|$$meeting$[{\sc Anon}]: {\sf BreakSimmetry}; $\vTtime> 4n^2$: {\sf Forward}} \label{p2:a2}
\Else
\State go to state {\sf Forward}
\EndIf
\AtState{Forward}
    \State \Call{Explore}{\dRight$|$ $marked \land \vEsteps >0$: {\sf Terminate}; $meeting$[{\sc Anon}]: {\sf BreakSimmetry}; $meeting$[{\sc MLeader}]: {\sf BeAvanguard}}

     \AtState{BreakSimmetry}
     \State if your pebble is on the node take it.
          \If{\#A=2}
      \State take a role in \{ \R, {\sc MLeader} \}.
	\State execute the \RT if \R or Alg.~\ref{alg:p2t} in state {\sf Go} if {\sc MLeader}. 
	\Else 
	\State  take a role in \{ \R, \Leader, \A \}.
	\State start algorithm \RT.
	\EndIf
	\AtState{BeAvanguard}
	     \State if your pebble is on the node take it.
	 \State take the role of \A.
	 \State  start algorithm \RT.
     \AtState{Terminate}

     \State Terminate \BH is the next node in clockwise direction.
\END
\end{algorithmic}
\caption{\GL; Phase 2 - Algorithm for scattered agents - {\sf Anon} \label{alg:p2chirscat}}
\end{algorithm*}

 \begin{algorithm*}
\footnotesize
\begin{algorithmic}[1]
\State  Predicates Shorthands: $NextUnsafe=\vEtime > \vMtime[C]$.
\State $FailedReport[\R]= \vMtime[C]> 2 ((\#Meets[\R]+1)+\vTsteps)$.

\State States: \{{\sf Go}, {\sf Cautious}, {\sf StartCP}, {\sf Terminate},  {\sf TerminateR}\}.
\AtState{Go}

    \State \Call{Explore}{\dRight $|$  $marked$: {\sf Cautious}; \pMeeting[Anon]={\sf StartCP}; $FailedReport[\R]$: {\sf TerminateR}; }
   \AtState{Cautious}

       \State \Call{Explore}{\dNil $|$ \pMeeting[Anon]: {\sf StartCP}; $NextUnsafe:$ {\sf Terminate}; $FailedReport[\R]$: {\sf TerminateR} }
            \AtState{StartCP}
            \State start algorithm \RT with the role of \Leader keeping the value of variable $\#Meets[\R]$.
     \AtState{Terminate}
      \State Terminate, \BH is in the next node in clockwise direction. 
           \AtState{TerminateR}
      \State Terminate, \BH is in the node that is at distance $\#Meets[\R]+1$ from counter-clockwise direction from the reference node. 
\END

\end{algorithmic}
\caption{\GL; Phase 2 - Algorithm for {\sc MLeader} \label{alg:p2t}}
\end{algorithm*}

\paragraph{Correctness of \GL.\label{sec:correctness}}

\begin{definition} (Gathered configuration)
We say that a group of $k$ agents {\em gathered} if either:
\begin{itemize}
\item There are $k$ agents on the same node; or,
\item There are $k-1$ agents on node $v_i$, and one agent $a$ on node $v_{i+1}$. Moreover, agent $a$ marked node $v_i$ with a pebble and has to still unmark it.
\end{itemize}
\end{definition}

Let us first start with a technical lemma, derived from~\cite{DiLFPPSV20}, and adapted to our specific case. 

\begin{lemma}\label{lemma:variationgathering}
If  $k$ agents perform a cautious walk in the same direction for an interval $I$ of $9n$ rounds and one of the agents does not explore $n$ nodes, then the agents gathered. 
\end{lemma}
\begin{proof}
Let $A$ be the set of agents performing a cautious walk, say in clockwise direction, and let $a^*$ be the agent that does not explore $n$ nodes. 

Agent $a^*$ can be blocked in progressing its cautious walk in two possible ways: (i) when it is trying to explore a new node by a missing edge in its clockwise direction (we say that $a^*$ is {\em forward blocked}); (ii) when
it is returning to a previously explored node to unmark it (it is blocked by an edge missing in its counter-clockwise direction, and we say that $a^*$ is {\em backward blocked}). 
If in a round $r$ an agent is forward blocked and another one is backward blocked, then they are on two endpoints of the same missing edge.  

 If $a^*$ is not blocked for $3(n-1)$ rounds then it would have explored $n$ nodes. Therefore, $a^*$ has been blocked for at least
 $6n-3$ rounds in an interval of $9n$ rounds.
If there is a round $r'$ when $a^*$ is blocked, then every $a\in A$
that at round $r'$ is not blocked does move (note that all blocked agents are either backward or forward blocked on the same edge of $a^*$).

Thus, all agents in $A$ that are not
in the same node as $a^*$
 move towards $a^*$ of at least  $\frac{6n-3}{3} =2n-1$ steps. 
On the other hand, every time  $a^*$ moves, the other agents might be blocked; however,  by hypothesis, this can occur less than $3n$ times.
 
 Since  the initial distance between $a^*$ and  an agent  in $A$  is at most $n-1$, it follows that
 such a distance increases less than $n-1$ (due to $a^*$ movements); however, it decreases by $2n-1$ (due to $a^*$ being blocked). In conclusion,
this distance is  zero: that is, they are either at the same node or at the two endpoints of a missing edge,  by the end of $I$, and the lemma follows.
 \color{black}
\end{proof}

\begin{lemma}\label{scattered:goodmarking}
Given three agents executing Phase 1, at most one of them enters the black hole. In this case, the counter-clockwise neighbour node of the black hole is marked by a pebble. 
\end{lemma}
\begin{proof}
If agents have not already met, then each agent performs a cautious walk, all in the same direction, marking a node and avoiding that other agents visit a possibly 
unsafe node (see state {\sf Init} in Algorithm~\ref{alg:scatteredchirality}): when the agent sees a marked node, it goes in state {\sf Waits}. In this state, the agent waits until it is sure that the next node is safe (that is, until the agent that marked the node returns to remove the pebble).  

When two agents meet, they become {\sc Follower} and {\sc Explorer}. By construction, {\sc Follower} never reaches \BH: in fact, {\sc Follower} moves a step clockwise only when it sees  {\sc Explorer} returning (see state {\sf Wait} and predicate $meeting[Explorer]$); this implies that the node where it moves is safe.
Also note that {\sc Explorer} never visits a possibly unsafe node if there is another agent on it: in fact, in state {\sf Explore}, there is a check on whether the current node is marked or not; if marked, {\sc Explorer} waits (thus, also blocking {\sc Follower}) until the next node can be deemed as safe.  

 If the third agent reaches {\sc Follower}, it will also become {\sc Follower} and it will never visit an unsafe node (recall that all agents move in the same direction, thus if {\sc Explorer} reached the black hole, any other agent has to first visit the node marked by {\sc Explorer}). Moreover, an {\sc Explorer} agent always marks a node before visiting its unexplored neighbour (see state {\sf Explore} of Algorithm~\ref{alg:scatteredchiralityexplorer}).

In conclusion, we have that at most one agent enters \BH, and the counter-clockwise neighbour node of \BH will be marked by a pebble, and the lemma follows. 
\end{proof}

\begin{observation}\label{scatter:correctterm}
If an agent terminates while executing Phase 1, then it terminates correctly
\end{observation}
\begin{proof}
The claim follows immediately by observing that the state {\sf Terminate} is always reached when an agent visits a marked node, the clockwise edge is not missing, and the agent that marked the node does not return. 
\end{proof}

\begin{lemma}\label{lemma:phase1}
Let us consider three agents executing Phase 1. If not all agents terminated locating the \BH, then Phase 1 ends by at most round $9n$ and, when it ends, we have one of the following:
\begin{itemize}
\item (1) all agents gathered; 
\item (2) at most one agent disappeared in the black hole, the counter-clockwise neighbour of the black hole is marked, and the remaining agents gathered;
\item (3) one agent terminated, the counter-clockwise neighbour of the black hole is marked, and the remaining agent has to still locate the \BH. 
\end{itemize}
\end{lemma}
\begin{proof}
By construction, in all states the agents check predicate $\vTtime = 9n$; thus, Phase 1 ends after at most $9n$ rounds. By Lemma~\ref{scattered:goodmarking} we have that at most one agent enters in \BH leaving its counter-clockwise neighbour marked.
Now we have three cases:
\begin{itemize}
\item One agent terminates, and by Observation~\ref{scatter:correctterm} it terminates correctly solving the \BHP.   The other agent has to still locate the \BH
\item One agent enters in the \BH and no one terminates. If no alive agent terminates, then it means that no one of them has explored $n$ nodes. Therefore, at the end of Phase 1 we have $\vTtime = 9n$ and by Lemma~\ref{lemma:variationgathering}  the agents gathered, and the lemma follows.
\item No one enters in the \BH and no one terminates. In this case we have that three agents gather by the end of Phase 1.  If agents end Phase 1 by predicate $\#A=3$, then the statement immediately follows.
Otherwise, $\vTtime = 9n$, by Lemma~\ref{lemma:variationgathering}  the agents gathered, and the lemma follows.
\end{itemize}
\end{proof}

The next lemma shows that, if \BH has been marked in Phase 1, then two agents executing Algorithm~\ref{alg:p2chirscat} solve \BHP in at most ${\cal O}(n^2)$ rounds. 
\begin{lemma}\label{lemma:twophase2chir}
Let us assume that the counter-clockwise neighbour $v$ of \BH has been marked by a pebble. If two agents executes Algorithm~\ref{alg:p2chirscat}, at least one of them terminates correctly locating the \BH in ${\cal O}(n^2)$ rounds; the other agent either
 terminates correctly locating the \BH or it never terminates.  
\end{lemma}
\begin{proof}
By Lemma~\ref{lemma:phase1}, at the first round of Phase 2 we have two possible cases:
\begin{itemize}
\item The two agents are at the same node. In this case, they immediately enter in state {\sf BreakSymmetry}. Let $a$ be the agent that takes the role of {\sc MLeader} and $b$ be the one that becomes \R. 
Their movements are similar to the ones of \Leader and \R in \RT, with the only difference that {\sc MLeader} moves until it reaches a marked node. 
By Lemma~\ref{lemma:phase1}, this marked node is the counter-clockwise neighbour of \BH; thus, if {\sc MLeader} reaches it, {\sc MLeader} correctly terminates. 

If {\sc MLeader} does not visit the marked node because of a missing edge, \R is able to move. By using  a similar argument to the one used in the proof of Theorem~\ref{th:pendulum}, the black hole is located in at most  ${\cal O}(n^2)$ rounds, and the lemma follows. 
 Also note that the only agent that can go in a termination state is {\sc MLeader}, therefore \R cannot terminate incorrectly. 

\item The two agents occupy two neighbouring nodes, and the most clockwise  agent does not have the pebble.
More precisely, agent $a$ is at node $v$, agent $b$ at node $v'$; also, agent $b$ is missing its pebble, and node $v$ is marked by a pebble. In this case, agent $a$  executes lines~\ref{p2:a1}-\ref{p2:a2} of Algorithm~\ref{alg:p2chirscat}: it removes the pebble from $v$, and waits for $4 \cdot n^2$ rounds.
Agent $b$  executes line~\ref{p2:b1} of Algorithm~\ref{alg:p2chirscat}: it moves towards node $v$ for $4 \cdot n^2$ rounds.  

If edge $e=(v, v')$ appears before the timeout, then $a$ and $b$ meet, and previous case applies. Otherwise, both agents go in state {\sf Forward}. In this state they both move clockwise. If one of them reaches the marked node, it correctly terminates. Otherwise the path towards \BH is blocked by a missing edge and the agents would meet in at most ${\cal O}(n)$ rounds.
When they meet, they both go in state {\sf BreakSymmetry}, and previous case applies again. Note that the above implies that at most one of the agents in state {\sf Forward} can be blocked forever by a missing edge. 
\end{itemize}
\end{proof}


\begin{lemma}\label{lemma:twophase2chir2}
Let us assume that three agents terminate Phase 1 gathered. Then, if three agents executes Algorithm~\ref{alg:p2chirscat}, at least one of them terminates correctly locating the \BH in ${\cal O}(n^2)$ rounds; the other agents either
 terminate correctly locating the \BH or never terminate.  
\end{lemma}
\begin{proof}
If three agents are on the same node, then they start \RT algorithm and the correctness follows from Theorem~\ref{th:pendulum}. 
Otherwise, we have two agents on a node $v$, with $v$ marked with a pebble, and the other agent $b$ on $v'$, with $v'$ the clockwise neighbour of $v$.
Upon the start of Phase 2, the two agents will become \R and {\sc MLeader}, respectively; {\sc MLeader} waits on the marked node, while $b$ tries to go
back to $v$. 

If edge $e=(v,v')$ is missing for $4 n^2$ rounds, then \R has enough time to reach the black hole, and {\sc MLeader} to terminate because of the fail to report of \R, hence the lemma follows. 
Note that after the termination of {\sc MLeader} the agent $b$, it goes in state {\sf Forward} it goes clockwise and, either enters the \BH or is blocked forever by a missing edge, in any case it
cannot terminate uncorrectly.

Thus, the last case to analyse is when \R is blocked by a missing edge in the first $4 n^2$ rounds. In this case, {\sc MLeader} and $b$ meet, and algorithm \RT starts. 
The lemma now follows by  Theorem~\ref{th:pendulum}.  
\end{proof}

\begin{lemma}\label{lemma:twophase2chir1}
Let us assume that a single agent $a$ starts  Phase 2. This agent executing Algorithm~\ref{alg:p2chirscat} either terminates correctly or it waits forever on a missing edge. 
\end{lemma}
\begin{proof}
By algorithm construction after at most $4 n^2$ rounds from the beginning of Phase 2 the agent $a$ goes in state {\sf Forward} and it starts moving in clockwise direction. Being the only agent still active,  it will
never change behaviour until it reaches the marked node or another agent.

 By Lemma~\ref{lemma:phase1} we have that the counter-clockwise neighbour of the \BH is marked, and the terminated agent is located at that node. 
 If no edge is removed forever $a$ will reach the marked node, and it will terminate. 
 Otherwise, $a$ will be forever blocked on a missing edge. In either case $a$ cannot terminate incorrectly. 
\end{proof}

\begin{theorem}\label{th:scat}
Given a dynamic ring ${\cal R}$, three anonymous agents with pebbles running \GL, solve \BHS in  ${\cal O}(n^2)$ moves and ${\cal O}(n^2)$ rounds.
\end{theorem}
\begin{proof}
By Lemma~\ref{lemma:phase1}, Phase 1 terminates in at most ${\cal O}(n)$ rounds. At this time, either: (1) \BHP is solved and all agents terminated, or  (2) the agents gathered, or  (3) the counter-clockwise neigbhour of \BH is marked and the remaining agents are gathered, or (4) there is still an agent active while an agent correctly terminated. 
In case (2), the proof follows by Lemma~\ref{lemma:twophase2chir2}.
In case (3), the proof follows by Lemma~\ref{lemma:twophase2chir}.
In case (4), we have just to show that the remaining agent does not terminate incorrectly. this is ensured by Lemma ~\ref{lemma:twophase2chir1}.
\end{proof}

By Th.~\ref{scat:lb} and Th.~\ref{th:scat} we have:
\begin{theorem}
Algorithm \GL is  size-optimal with optimal cost and time. 
\end{theorem} 
\bibliography{references}

\begin{thebibliography}{10}

\bibitem{AaKM14}
E.~Aaron, D.~Krizanc, and E.~Meyerson.
\newblock {DMVP}: Foremost waypoint coverage of time--varying graphs.
\newblock In {\em 40th International Workshop on Graph--Theoretic Concepts in
  Computer Science}, pages 29--41, 2014.

\bibitem{AbsM14}
S.~Abshoff and F.~Meyer auf~der Heide.
\newblock Continuous aggregation in dynamic ad-hoc networks.
\newblock In {\em 21st Int. Coll. on Structural Inf. and Comm. Compl.}, pages
  194--209, 2014.

\bibitem{Agarwalla2018}
A.~Agarwalla, J.~Augustine, W.~Moses, S.~Madhav, and A.K. Sridhar.
\newblock Deterministic dispersion of mobile robots in dynamic rings.
\newblock In {\em 19th International Conference on Distributed Computing and
  Networking}, pages 19:1--19:4, 2018.

\bibitem{balamohan2014exploring}
B.~Balamohan, S.~Dobrev, P.~Flocchini, and N.~Santoro.
\newblock Exploring an unknown dangerous graph with a constant number of
  tokens.
\newblock {\em Theoretical Computer Science}, 2014.

\bibitem{BournatDD16}
M.~Bournat, A.K. Datta, and S.~Dubois.
\newblock Self--stabilizing robots in highly dynamic environments.
\newblock In {\em 18th International Symposium on Stabilization, Safety, and
  Security of Distributed Systems}, pages 54--69, 2016.

\bibitem{BournatDP17}
M.~Bournat, S.~Dubois, and F.~Petit.
\newblock Computability of perpetual exploration in highly dynamic rings.
\newblock In {\em 37th {IEEE} International Conference on Distributed Computing
  Systems}, pages 794--804, 2017.

\bibitem{BournatDP18}
M.~Bournat, S.~Dubois, and F.~Petit.
\newblock Gracefully degrading gathering in dynamic rings.
\newblock In {\em 20th International Symposium on Stabilization, Safety, and
  Security of Distributed Systems}, pages 349--364, 2018.

\bibitem{CaFMS14}
A.~Casteigts, P.~Flocchini, B.~Mans, and N.~Santoro.
\newblock Measuring temporal lags in delay-tolerant networks.
\newblock {\em IEEE Transactions on Computer}, 63(2):397--410, 2014.

\bibitem{CaFQS12}
A.~Casteigts, P.~Flocchini, W.~Quattrociocchi, and N.~Santoro.
\newblock Time-varying graphs and dynamic networks.
\newblock {\em Int. J. Parallel Emergent Distributed Syst.}, 27(5):387--408,
  2012.

\bibitem{ChDLM11b}
J.~Chalopin, S.~Das, A.~Labourel, and E.~Markou.
\newblock Black hole search with finite automata scattered in a synchronous
  torus.
\newblock In {\em Proc of 25th International Symposium on Distributed
  Computing}, pages 432--446, 2011.

\bibitem{ChDLM13}
J.~Chalopin, S.~Das, A.~Labourel, and E.~Markou.
\newblock Tight bounds for black hole search with scattered agents in a
  synchronous ring.
\newblock {\em Theoretical Computer Science}, 509:70--85, 2013.

\bibitem{CzKMP06}
J.~Czyzowicz, D.~Kowalski, E.~Markou, and A.~Pelc.
\newblock Complexity of searching for a black hole.
\newblock {\em Fundamenta Informaticae}, 71:229--242, 2006.

\bibitem{CzKMP07}
J.~Czyzowicz, D.~Kowalski, E.~Markou, and A.~Pelc.
\newblock Searching for a black hole in synchronous tree networks.
\newblock {\em Combinatorial Probabilistic Computing}, 16(4):595--619, 2007.

\bibitem{DaDPP19}
S.~Das, G.A. {Di Luna}, L.~Pagli, and G.~Prencipe.
\newblock Compacting and grouping mobile agents on dynamic rings.
\newblock In {\em 15th Annual Conference on Theory and Applications of Models
  of Computation}, pages 114--133, 2019.

\bibitem{d2013exploring}
M.~D'Emidio, D.~Frigioni, and A.~Navarra.
\newblock Exploring and making safe dangerous networks using mobile entities.
\newblock In {\em Ad--hoc, Mobile, and Wireless Network}, pages 136--147.
  Springer, 2013.

\bibitem{DiL19}
G.A. {Di Luna}.
\newblock {\em Mobile Agents on Dynamic Graphs, {\em Chapter 20 of
  \cite{FlPS19}}}.
\newblock Springer, 2019.

\bibitem{briefdiluna}
G.A. {Di Luna} and R.~Baldoni.
\newblock Brief announcement: Investigating the cost of anonymity on dynamic
  networks.
\newblock In {\em 34th Symposium on Principles of Distributed Computing}, pages
  339--341, 2015.

\bibitem{DiLDFS20}
G.A. {Di~Luna}, S.~Dobrev, P.~Flocchini, and N.~Santoro.
\newblock Distributed exploration of dynamic rings.
\newblock {\em Distributed Computing}, 33:41--67, 2020.

\bibitem{DiLFPPSV20}
G.A. {Di Luna}, P.~Flocchini, L.~Pagli, G.~Prencipe, N.~Santoro, and
  G.~Viglietta.
\newblock Gathering in dynamic rings.
\newblock {\em Theoretical Computer Science}, 811:79--98, 2020.

\bibitem{dobrev2013exploring}
S.~Dobrev, P.~Flocchini, R.~Kr{\'a}lovi{\v{c}}, and N.~Santoro.
\newblock Exploring an unknown dangerous graph using tokens.
\newblock {\em Theoretical Computer Science}, 472:28--45, 2013.

\bibitem{Dobrev2002}
S.~Dobrev, P.~Flocchini, G.~Prencipe, and N.~Santoro.
\newblock Searching for a black hole in arbitrary networks: optimal mobile
  agents protocols.
\newblock {\em Distributed Computing}, 19(1):1--35, 2006.

\bibitem{DobrevFPS07}
S.~Dobrev, P.~Flocchini, G.~Prencipe, and N.~Santoro.
\newblock Mobile search for a black hole in an anonymous ring.
\newblock {\em Algorithmica}, 48(1):67--90, 2007.

\bibitem{dobrev2007locating}
S.~Dobrev, N.~Santoro, and W.~Shi.
\newblock Locating a black hole in an un--oriented ring using tokens: The case
  of scattered agents.
\newblock In {\em 13th International Euro--Par Conference European Conference
  on Parallel and Distributed Computing}, pages 608--617. Springer, 2007.

\bibitem{4228188}
S.~{Dobrev}, N.~{Santoro}, and W.~{Shi}.
\newblock Scattered black hole search in an oriented ring using tokens.
\newblock In {\em 25th IEEE International Parallel and Distributed Processing
  Symposium}, pages 1--8, 2007.

\bibitem{ErHK15}
T.~Erlebach, M.~Hoffmann, and F.~Kammer.
\newblock On temporal graph exploration.
\newblock In {\em 42nd International Colloquium on Automata, Languages, and
  Programming}, pages 444--455, 2015.

\bibitem{ErS20}
T.~Erlebach and J.T. Spooner.
\newblock A game of cops and robbers on graphs with periodic edge-connectivity.
\newblock In {\em 46th International Conference on Current Trends in Theory and
  Practice of Informatics}, pages 64--75, 2020.

\bibitem{FIS12}
P.~Flocchini, D.~Ilcinkas, and N.~Santoro.
\newblock Ping pong in dangerous graphs: optimal black hole search with
  pebbles.
\newblock {\em Algorithmica}, 62(3-4):1006--1033, 2012.

\bibitem{Flocchini2009}
P.~Flocchini, M.~Kellett, P.~Mason, and N.~Santoro.
\newblock Map construction and exploration by mobile agents scattered in a
  dangerous network.
\newblock In {\em 28th IEEE International Parallel and Distributed Processing
  Symposium}, pages 1--10, 2009.

\bibitem{FlKMS12}
P.~Flocchini, M.~Kellett, P.~Mason, and N.~Santoro.
\newblock Searching for black holes in subways.
\newblock {\em Theory of Computing Systems}, 50(1):158--184, 2012.

\bibitem{FlMS13}
P.~Flocchini, B.~Mans, and N.~Santoro.
\newblock On the exploration of time-varying networks.
\newblock {\em Theoretical Computer Science}, 469:53--68, 2013.

\bibitem{FlPS19}
P.~Flocchini, G.~Prencipe, and N.~{Santoro {\em (Eds.)}}.
\newblock {\em Distributed Computing by Mobile Entities}.
\newblock Springer, 2019.

\bibitem{GoFMS19}
T.~Gotoh, P.~Flocchini, T.~Masuzawa, and N.~Santoro.
\newblock Tight bounds on distributed exploration of temporal graphs.
\newblock In {\em 23rd International Conference on Principles of Distributed
  Systems}, pages 22:1--22:16, 2019.

\bibitem{HaeK12}
B.~Haeupler and F.~Kuhn.
\newblock Lower bounds on information dissemination in dynamic networks.
\newblock In {\em 26th Int. Symp. on Distributed Computing}, pages 166--180,
  2012.

\bibitem{IlKW14}
D.~Ilcinkas, R.~Klasing, and A.M. Wade.
\newblock Exploration of constantly connected dynamic graphs based on cactuses.
\newblock In {\em 21st International Colloquium Structural Information and
  Communication Complexity}, pages 250--262, 2014.

\bibitem{IlcinkasW18}
D.~Ilcinkas and A.M. Wade.
\newblock Exploration of the {T}--interval--connected dynamic graphs: the case
  of the ring.
\newblock {\em Theory of Computing Systems}, 62(5):1144--1160, 2018.

\bibitem{JatYG14}
R.~Jathar, V.~Yadav, and A.~Gupta.
\newblock Using periodic contacts for efficient routing in delay tolerant
  networks.
\newblock {\em Ad Hoc \& Sensor Wireless Networks}, 22(1,2):283--308, 2014.

\bibitem{KMRS07}
R.~Klasing, E.~Markou, T.~Radzik, and F.~Sarracco.
\newblock Hardness and approximation results for black hole search in arbitrary
  networks.
\newblock {\em Theor. Comput. Sci.}, 384(2-3):201--221, 2007.

\bibitem{KMRS08}
R.~Klasing, E.~Markou, T.~Radzik, and F.~Sarracco.
\newblock Approximation bounds for black hole search problems.
\newblock {\em Networks}, 52(4):216--226, 2008.

\bibitem{KuhLoO11}
F.~Kuhn, T.~Locher, and R.~Oshman.
\newblock Gradient clock synchronization in dynamic networks.
\newblock {\em Theory of Computing Systems}, 49(4):781--816, 2011.

\bibitem{KuhLyO10}
F.~Kuhn, N.~Lynch, and R.~Oshman.
\newblock Distributed computation in dynamic networks.
\newblock In {\em 42nd ACM Symposium on Theory of Computing}, pages 513--522,
  2010.

\bibitem{m2020live}
S.~Mandal, A.~Rahaman Molla, and W.K.~Moses Jr.
\newblock Live exploration with mobile robots in a dynamic ring, revisited,
  2020.
\newblock \href {http://arxiv.org/abs/2001.04525} {\path{arXiv:2001.04525}}.

\bibitem{markou2012identifying}
E.~Markou.
\newblock Identifying hostile nodes in networks using mobile agents.
\newblock {\em Bulletin of the European Association for Theoretical Computer
  Science}, 108:93--129, 2012.

\bibitem{opodis12}
E.~Markou and M.~Paquette.
\newblock Black hole search and exploration in unoriented tori with synchronous
  scattered finite automata.
\newblock In {\em 14th International Conference on Principles of Distributed
  Systems}, pages 239--253, 2012.

\bibitem{MaS19}
E.~Markou and W.~Shi.
\newblock {\em Dangerous Graphs, {\em Chapter 18 of \cite{FlPS19}}}.
\newblock Springer, 2019.

\bibitem{Michail2014}
O.~Michail and P.G. Spirakis.
\newblock Traveling salesman problems in temporal graphs.
\newblock {\em Theoretical Computer Science}, 634:1--23, 2016.

\bibitem{OdW05}
R.~O'Dell and R.Wattenhofer.
\newblock Information dissemination in highly dynamic graphs.
\newblock In {\em Joint Workshop on Foundations of Mobile Computing}, pages
  104--110, 2005.

\bibitem{shi2018Token}
M.~Peng, W.~Shi, and J.~Corriveau.
\newblock Repairing faulty nodes and locating a dynamically spawned black hole
  search using tokens.
\newblock In {\em 6th IEEE Conference on Communications and Network Security},
  pages 136--145. IEEE, 2018.

\end{thebibliography}

\end{document}